\documentclass[journal,twoside,web]{ieeecolor}
\usepackage{generic}
\usepackage{cite}
\usepackage{color}

\definecolor{purple}{rgb}{0.75,0.194,0.34}
\usepackage{amsmath,amssymb,amsfonts}
\usepackage{algorithmic}
\usepackage{graphicx}
\usepackage{algorithm,algorithmic}

\makeatletter
\let\NAT@parse\undefined
\makeatother

\usepackage[colorlinks=true,
    linkcolor=blue,
    citecolor=blue]{hyperref}

\newtheorem{theo}{Theorem}
\newtheorem{assu}{Assumption}
\newtheorem{lemma}{Lemma}

\newtheorem{remark}{Remark}

\newtheorem{crit}{Criterion}

\usepackage{textcomp}
\def\BibTeX{{\rm B\kern-.05em{\sc i\kern-.025em b}\kern-.08em
    T\kern-.1667em\lower.7ex\hbox{E}\kern-.125emX}}
\begin{document}
\title{Output Feedback MPC with Adaptive Tubes}
\author{Anchita Dey*, \IEEEmembership{Graduate Student Member, IEEE}, and Shubhendu Bhasin
\thanks{This paragraph of the first footnote will contain the date on 
which you submitted your paper for review. }
\thanks{Anchita Dey and Shubhendu Bhasin are with the department of Electrical Engineering, Indian Institute of Technology Delhi, Hauz Khas, New Delhi, Delhi 110016, India {\tt\small anchitadey.ee.india@gmail.com, sbhasin@ee.iitd.ac.in.}}
\thanks{*Corresponding author}
}

\maketitle
\begin{abstract}
An output feedback model predictive control (MPC) framework with \textit{adaptive tubes} is proposed for linear time-invariant systems subject to parametric and additive uncertainties. An adaptive observer provides point estimates of the system state, model parameters, and initial condition, while jointly updating the corresponding sets containing the true parameters and initial state. These estimates parameterize the constrained optimal control problem, enabling constraint tightening, terminal ingredients, and tube geometry to be updated as the estimates evolve. In contrast to standard robust tube-based MPC formulations, the proposed approach does not require a common quadratically stabilizing linear feedback gain across the parametric uncertainty set. As the available uncertainty information improves, the tube geometry evolves accordingly, resulting in an adaptive tube MPC framework with improved performance over time. Recursive feasibility and robust exponential stability are established, and a numerical example is presented.
\end{abstract}
\begin{IEEEkeywords}
Model predictive control (MPC), Uncertain systems, Output feedback and Observers, Constrained control
\end{IEEEkeywords}

\section{Introduction}

Model predictive control (MPC) is a well-known technique for optimal control of systems subject to hard constraints on states and inputs. Classical MPC \cite{kouvaritakis2016model,Rawlings2017MPC} relies on a system model to predict future trajectories and compute an optimal control sequence over a finite horizon. The resulting constrained optimal control problem (COCP) requires accurate model knowledge to ensure feasibility and stability. In practice, however, model uncertainty and incomplete state measurements limit the direct applicability of such formulations. Both limitations must be addressed simultaneously in any practically deployable MPC framework.

Robust MPC methods are predominantly formulated either as min-max optimization or LMI-based schemes, or within the robust tube-based framework. These approaches are typically designed for systems with known parameters subject to additive disturbances \cite{CHISCI20011019,langson2004robust,rakovic2012homothetic}, or for systems with parametric uncertainty, with or without additive disturbances \cite{fleming2015,KOTHARE19961361,buj2021}. In the latter case, however, no attempt is made to estimate or learn the uncertain parameters. To mitigate conservatism and enhance performance, adaptive MPC has gained attention \cite{jafari2017adaptive,dhar2021indirect,kohlerrampc,lorenzen2019robust,lu2023robust,anchita2025,anch}. In these formulations, uncertain parameters, or the sets containing them, or both, are updated online, while robustness to estimation errors is maintained through tube-based constructions. Homothetic tube parameterizations \cite{langson2004robust,rakovic2012homothetic} are particularly attractive due to their tractability. Recent works have also explored refined tube representations; \cite{luflexitube} proposes flexible tubes based on zonotopes, where the tube cross-sectional parameters are optimized online. The framework in \cite{kogelblending} considers systems with additive disturbances and constructs tubes via online blending of precomputed sets associated with different stabilizing feedback gains. Self-tuning tube schemes for systems with both parametric and additive uncertainties are developed in \cite{tranosrusso}. The recent work in \cite{wang2025tube} refines classical rigid tube MPC \cite{langson2004robust} for LTI systems with additive disturbances by introducing time-varying tube cross-sections within a state feedback design framework, while assuming known system parameters.

The MPC frameworks for systems with parametric uncertainty and state constraints, such as \cite{buj2021,jafari2017adaptive,dhar2021indirect,lorenzen2019robust,lu2023robust,tranosrusso,kohlerrampc,luflexitube,anchita2025,anch}, typically assume the existence of a quadratically stabilizing linear state feedback gain valid over the entire uncertainty set \cite{khargonekar1990robust}. While this assumption enables tractable design, it restricts the admissible uncertainty sets \cite{Gahinet1994,barmish1985necessary}, and such a common gain may not exist for all parameter realizations \cite{Petersen1985}. Data-driven approaches \cite{berberich2020data,coulson2019data} avoid this requirement, but rely on offline representations that may be sensitive to changes in system dynamics and typically impose constraints on inputs and outputs rather than all states.

In this paper, an output feedback MPC framework with adaptive tubes is proposed for discrete-time LTI systems subject to parametric and additive uncertainties. Similar to \cite{anchita2025}, an adaptive observer is employed together with a reformulated COCP, with two key modifications. First, point estimates of the initial condition are obtained in addition to the state and parameter estimates. Second, available data is used to jointly update the uncertainty sets for the parameters and the initial condition using set membership-based identification \cite{boydset}. Further, the linear regression for parameter adaptation is augmented by combining past and current data, enabling improved convergence \cite{kreisselmeier1977adaptive,lion1967rapid,dey2025initialexcitationbasedadaptiveobservers}. To formulate the COCP, the system dynamics is expressed using the current estimates, with the resulting mismatch treated as a lumped additive disturbance. Based on this representation, constraint tightening, terminal set computation, and tube design are carried out as in existing robust tube MPC \cite{rakovic2012homothetic,CHISCI199815}. The COCP components are designed for the current parameter estimate, and therefore require only a stabilizing feedback gain, terminal cost, and terminal set for that estimate, eliminating the need for a common design across all admissible parametric uncertainties.

As the estimates and uncertainty sets evolve, the tube geometry and the COCP components change accordingly, leading to an \emph{adaptive tube} MPC framework. While contraction of the uncertainty sets preserves recursive feasibility and stability, the central technical challenge is ensuring that the time-varying COCP retains its stability and feasibility properties despite changes in the prediction model induced by point estimate updates. The proposed solution avoids the standard requirement of a common quadratically stabilizing linear feedback gain over the entire uncertainty set. Instead, a compatibility condition is imposed on the successive terminal ingredients that ensures that the Lyapunov function remains non-increasing along the closed-loop trajectories---a condition that is strictly less restrictive than the quadratic stabilization over the entire uncertainty set. Recursive feasibility is maintained using a backup solution from the previous time step. Improved uncertainty characterization results in tighter prediction bounds, and appending the state estimation error sets to the homothetic tube returned by the COCP yields an outer tube containing the true state trajectory. 

The contributions are summarized as follows:
\begin{itemize}
    \item An adaptive observer design that estimates the initial state along with system parameters and state, and jointly updates the associated uncertainty sets using set membership-based identification.
    \item An MPC formulation in which the COCP components, including constraint tightening, terminal set, and tube geometry, are updated based on the available estimates.
    \item Recursive feasibility and robust exponential stability guarantees for the proposed adaptive tube MPC framework, without requiring a common quadratically stabilizing linear feedback gain.
\end{itemize}

\textit{Notations:} $\oplus$ denotes Minkowski sum defined as $\mathbb{P}\oplus\mathbb{Q}\triangleq \{p+q\;|\;p\in\mathbb{P},\;q\in\mathbb{Q}\}$ and $p\oplus\mathbb{Q}\triangleq \{p+q\;|\; q\in\mathbb{Q}\}$. $\ominus$ denotes Pontryagin difference defined as $\mathbb{P}\ominus\mathbb{Q}\triangleq\{p\;|\;p+q\in\mathbb{P}\;\forall q\in\mathbb{Q}\}$. The multiplication of any matrix $V$ with a set $\mathbb P$ is defined as $V\mathbb P\triangleq \{ Vp\;|\;\forall p\in\mathbb P  \} $, where $V$ and $p$ are of conforming dimensions. For any matrix $V$, $V(:,p:q)$ and $V(r:s,p:q)$ represent the sub-matrices made of all rows and $r$-th to $s$-th rows, respectively, with $p$-th to $q$-th columns. The convex hull of all elements in $\mathbb{P}$ is denoted by $\mathrm{\mathbf{co}} (\mathbb{P})$. The notation $V\succ 0$ $(\succeq 0)$ represents that the matrix $V$ is symmetric positive-definite (positive semi-definite). $\|\cdot\|_\infty$ and $\|\cdot\|_2$ represent $\infty$ and $2$-norms, of a vector (or induced $\infty$ and $2$-norms of a matrix), respectively. $\|p\|^2_V\triangleq p^\top V p$ for a vector $p$. The identity matrix and zero matrix are denoted by $I_n\in\mathbb{R}^{n\times n}$ and $0_{m\times n}\in\mathbb{R}^{m\times n}$, respectively. $\{\alpha(i)\}_{i=p:q}$ represents the set $\{\alpha(p), \alpha(p+1),...,\alpha(q)\}$, where $\alpha(i)$ is a function of $i$, and $p$, $q>p$ are integers. $\mathbb{I}_p^q$ represents the set of all integers from $p$ to $q$. The value of $p$ at time $t+i$ predicted at time $t$ is denoted by $p_{i|t}$. $(\cdot)^*$ denotes the optimal value returned by solving the COCP. Any signal belonging to $\mathcal L_\infty$ implies it is bounded. The operation $\mathrm{\mathbf{vec}}(\cdot)$ on a matrix $W=[ W_1^\top  \;\; W_2^\top  \;...\;W_p^\top ]^\top \in\mathbb{R}^{p\times q}$, where $W_i^\top \in\mathbb{R}^q$ $\forall i\in\mathbb{I}_1^p$ is defined as $\mathrm{\mathbf{vec}}(W)\triangleq 
   [ W_1 \;\; W_2 \;...\;W_p]^\top \in\mathbb{R}^{pq}$, and $\mathrm{\mathbf{vec}}^{-1}(\cdot,\cdot)$ is the inverse operation of $\mathrm{\mathbf{vec}}(\cdot)$, i.e., $\mathrm{\mathbf{vec}}^{-1}\left( [ W_1 \;\; W_2 \;...\;W_p]^\top ,q  \right)\triangleq W$.

\section{Problem Statement}
The objective is to stabilize the origin of a discrete-time LTI system in the following observable canonical realization:

\begin{subequations}\label{sys1}
    \begin{align}
  &  x_{t+1}= \underbrace{\left[\begin{array}{c|c} {\mathcal{A}}\;&\begin{array}{cc}
         I_{n-q}\\0_{q\times (n-q) }  \end{array} \end{array} \right]}_{=:A} x_t+Bu_t+d_t \label{sys1x}\\
  &  y_t= \underbrace{\begin{bmatrix}    I_q & 0_{q\times (n-q)}   \end{bmatrix}}_{=:C} x_t \label{sys1y}
\end{align}
\end{subequations}
subject to hard constraints
\begin{align}
    x_t\in\mathbb{X},\;u_t\in\mathbb U\quad \forall t\in\mathbb I_0^\infty. \label{hc}
\end{align}
Here, $x_t\in\mathbb{R}^n$, $u_t\in\mathbb{R}^m$, $y_t\in\mathbb R^q$, and $d_t\in\mathbb{R}^n$ denote the state, input, output, and disturbance at time $t$, respectively. The disturbance $d_t$ and the parameters $\mathcal A\in\mathbb R^{n\times q}$, $B\in\mathbb R^{n\times m}$ are unknown. Only output measurements $y_t$ are available. The constraint sets $\mathbb X$ and $\mathbb U$ are known convex polytopes containing their respective origins in the interior.

To ensure robust feasibility and stability, it is standard to assume that uncertainties belong to known bounded sets \cite{kouvaritakis2016model,Rawlings2017MPC,langson2004robust,CHISCI20011019,fleming2015,KOTHARE19961361,buj2021,jafari2017adaptive,kohlerrampc,kogelblending,tranosrusso,wang2025tube,lorenzen2019robust,dhar2021indirect,anchita2025,anch}. For the problem of interest, the additive and parametric uncertainties are assumed to lie in known convex bounded sets, with $d_t\in\mathbb D\subseteq\mathbb R^n$ and $\psi\triangleq\begin{bmatrix}
    A & B
\end{bmatrix}\in\Psi\subseteq\mathbb R^{n\times (n+m)}$. The set $\mathbb D$ is a polytope containing the origin in its interior. The set $\Psi$ is defined as
\begin{align}
    \Psi\triangleq \mathrm{\mathbf{co}}\left(\psi^{[1]},\;\psi^{[2]},...,\psi^{[L_p]}\right),
\end{align}
where $L_p$ is a known finite positive integer. For each $\hat\psi\triangleq\begin{bmatrix}
        \hat A & \hat B
    \end{bmatrix}\in\Psi$, there exists a pair $(P,K)$ specific to $\hat \psi$ such that $P\succ 0$ and
    \begin{align}\label{originaldare}
        P-(\hat A+\hat B K)^\intercal P (\hat A +\hat B K)-Q - K^\intercal R K \succeq 0,
    \end{align}
for given weighting matrices $Q, R \succ 0$, with $P, Q \in \mathbb{R}^{n \times n}$, $R \in \mathbb{R}^{m \times m}$, and $K \in \mathbb{R}^{m \times n}$. This condition ensures that each system corresponding to $\hat\psi \in \Psi$ is stabilizable and admits a quadratic Lyapunov function of the form $x^\top P x$ under a suitable linear state feedback gain $K=-(R + B^\top P B)^{-1} B^\top P A$, and is a standard assumption in MPC formulations \cite{kouvaritakis2016model,rakovic2012homothetic,langson2004robust}. Unlike standard robust tube MPC formulations that require a common quadratically stabilizing feedback gain over the entire uncertainty set \cite{buj2021,jafari2017adaptive,dhar2021indirect,lorenzen2019robust,lu2023robust,tranosrusso,kohlerrampc,luflexitube,anchita2025,anch}, condition \eqref{originaldare} only assumes existence of a stabilizing pair for each admissible parameter realization individually\footnote{{This distinction is exploited later in the proposed adaptive tube MPC framework.}}. 

Further, a known convex polytope $\mathbb X_0\subseteq \mathbb X$ with $L_x$ vertices is considered for the initial state $x_0$, where $L_x$ is a known finite positive integer.

In the ideal case, when the states and parameters are known and no disturbance is present, the origin of \eqref{sys1} can be stabilized by implementing the classical MPC COCP given below.
\begin{align*}
\text{ }\min_{\{u_{i|t}\}_{i=0:N-1}} &\sum_{i=0}^{N-1}\left( \|x_{i|t}\|^2_{{Q}}+\|u_{i|t}\|_R^2\right)+\|x_{N|t}\|^2_{{P}}\\
\text{subject to }&x_{0|t}=x_t,\\
 &x_{i+1|t}=Ax_{i|t}+Bu_{i|t}\;\;\;\;\forall i\in\mathbb{I}_{0}^{N-1},\\
&x_{i|t}\in\mathbb{X}\hspace{0.2cm}\forall i\in\mathbb{I}_0^{N},\;\;u_{i|t}\in\mathbb{U}\hspace{0.2cm}\forall i\in\mathbb{I}_{0}^{N-1},\\
 &x_{N|t}\in\mathbb{X}_{TS}\subseteq\mathbb{X},
\end{align*}
where $N$ is the prediction horizon length, and $\mathbb{X}_{TS}$ is the terminal set \cite[Ch.~2]{kouvaritakis2016model}. The set $\mathbb{X}_{TS}$ is generally constructed to be the maximal positive invariant set for the dynamics $z_{t+1}=(A+BK)z_t$ with admissibility to constraints $\mathbb{X}_{TS}\subseteq \mathbb X$ and $K\mathbb{X}_{TS}\subseteq \mathbb U$. 

Due to parameter uncertainty, external disturbance and unavailability of full state information, this COCP cannot be directly implemented. An alternative is to estimate these quantities online and reformulate the COCP to account for estimation errors and disturbances, as explored in \cite{anchita2025,anch}.

However, these approaches, along with several tube-based MPC formulations for systems with parametric uncertainty \cite{buj2021,jafari2017adaptive,dhar2021indirect,lorenzen2019robust,lu2023robust,tranosrusso,kohlerrampc,luflexitube}, rely on the existence of a quadratically stabilizing linear feedback gain over the uncertainty set. This requirement stems from the use of a single COCP designed to accommodate all admissible uncertainties.

The objective of this work is therefore to develop an output feedback MPC framework for stabilizing the origin of \eqref{sys1}, while maintaining recursive feasibility, robust constraint satisfaction of \eqref{hc}, and stability, without requiring a common quadratically stabilizing linear feedback gain over the entire uncertainty set. To this end, an adaptive observer is employed to obtain online point and set estimates of the uncertain parameters and initial condition. These updated uncertainty descriptions are subsequently incorporated into the COCP through the constraint tightening, terminal ingredients, and tube construction. As the uncertainty sets and point estimates evolve with incoming data, the associated tube geometry also changes accordingly, leading to an adaptive tube MPC framework with progressively reduced conservatism. The adaptive observer is designed in the following section.

\section{Adaptive observer with set-membership identification}
The plant dynamics in \eqref{sys1x} can be rewritten as
\begin{align}
    x_{t+1}&=Fx_t+(A-F)x_t+Bu_t+d_t,\nonumber\\
   & =F x_t + \begin{bmatrix}
       Y_t & U_t
   \end{bmatrix}\underbrace{\begin{bmatrix}
       (a-f)^\intercal & b^\intercal
   \end{bmatrix}^\intercal}_{=:p}+d_t,
\end{align}
where $p\in\mathbb R^{qn+mn}$, $Y_t\triangleq I_n\otimes y_t^\intercal\in\mathbb R^{n\times qn}$, $U_t\triangleq I_n\otimes u_t^\intercal\in\mathbb R^{n\times mn}$, and $F\in\mathbb R^{n\times n}$ is a user-defined Schur stable matrix structurally similar to $A$ in \eqref{sys1x} given by, 
\begin{align}
  & F\triangleq \left[
    \begin{array}{c|c}
   {\mathcal{F}}\;&\begin{array}{cc}
         I_{n-q}\\0_{q\times (n-q) }
    \end{array}
    \end{array}\right],\text{ where }\mathcal F\in\mathbb R^{n\times q}\label{obsF}\\
   \text{and }& a\triangleq \mathrm{\mathbf{vec}} \left( \mathcal A \right)\in\mathbb R^{qn},\;b\triangleq \mathrm{\mathbf{vec}}(B)\in\mathbb R^{mn},\\
   & f\triangleq \mathrm{\mathbf{vec}} \left( \mathcal F \right)\in\mathbb R^{qn}.\label{obsF2}
\end{align}

The mapping from $\psi$ to $p$ for a given $f$ is a bijection, and therefore, it suffices to estimate $p$, which belongs to the set 
\begin{align}\label{PiSetdefn}
    \Pi\triangleq \left\{ \mathcal Z_1 \left(\bar\psi, f, q \right) \;\middle|\; \forall \bar\psi\in\Psi ,\;f\text{ in \eqref{obsF2}} \right\}.
\end{align}
The explicit definitions of the bijective mapping $\mathcal Z_1$ and its inverse are provided in Appendix \ref{app:z1z1inv} for completeness. The inverse mapping is used later to recover the estimate of $\psi$ (and therefore of $A$ and $B$) from the estimate of $p$.

\subsection{Formulation of filter-based regression}

To estimate $p$ and $x_t$ using $u_t$ and $y_t$, filter matrices $M_t\in\mathbb R^{n\times (qn+mn)}$ are defined to construct a linear regression. The filter variables are computed using the dynamics
\begin{align}
    M_{t+1}=FM_t+\begin{bmatrix}
        Y_t & U_t
    \end{bmatrix},\;\;M_0=0_{n\times (qn+mn)}.\label{FilterM}
\end{align}
The plant state and output can then be expressed as 
\begin{subequations}
\begin{align}
x_t&=M_tp+F^tx_0+\sum_{k=0}^{t-1}F^{t-1-k}d_k,\label{xtrue}\\
y_t&=CM_t p+CF^t x_0 + \underbrace{C\sum_{k=0}^{t-1}F^{t-1-k}d_k}_{=:\eta_t},\label{ytrueobs}
\end{align}
\end{subequations}
respectively. Let $\hat p_t$ and $\hat x_{0_t}$ be the estimates of $p$ and $x_0$, respectively, at time $t$. The adaptive observer state and output are then defined as
\begin{subequations}\label{adapobs1}
\begin{align}
\hat x_t&=M_t\hat p_t+F^t \hat x_{{0_t}},  \label{xhat}\\
\hat y_t&=CM_t \hat p_t+CF^t \hat x_0=\underbrace{\begin{bmatrix}
     CM_t & CF^t
 \end{bmatrix}}_{=:\omega_t} \underbrace{\begin{bmatrix}
     \hat p_t\\ \hat x_{0_t}
 \end{bmatrix}}_{=:\hat\theta_t}, \label{yhat}
\end{align}
\end{subequations}
respectively, where $\omega_t\in\mathbb R^{q \times (qn+mn+n)}$ is the regressor, and $\hat \theta_t\in\mathbb R^{qn+mn+n}$ is the estimate of
\begin{align}
    \theta\triangleq \begin{bmatrix}
    p^\intercal & x_0^\intercal
\end{bmatrix}^\intercal\in\mathbb R^{qn+mn+n}.
\end{align}

Since $x_0\in\mathbb X_0$, the choice of the initial state estimate $\hat x_{0_t}$ is restricted to $\mathbb X_0$. Let the set containing the initial state estimation error $\tilde{x}_{0_t}$ be given by
\begin{align}\label{X0ttildeset0}
    \widetilde{\mathbb X}_{0_t}\triangleq\{x\;|\;x+\hat x_{0_t}\in \mathbb X_0\}.
\end{align}
Then, $\widetilde{\mathbb X}_{0_0}$ is a known convex polytope containing the origin. In addition, for characterizing bounded sets for the state estimation errors $\tilde x_t\triangleq x_t-\hat x_t$, the following standard assumption is made \cite{anchita2025,anch,kogel2017robust,mayne2006robust,mayne2009robust}.

\begin{assu}\label{Assum:1}
    The set $\widetilde{\mathbb X}_{0_0}\triangleq \mathbb{X}_0\ominus \hat x_{0_0}$ is positive invariant to the dynamics $z_{t+1}=Fz_t$, i.e., 
    \begin{align}\label{X0inv}
        F\widetilde{\mathbb X}_{0_0}\subseteq \widetilde{\mathbb X}_{0_0}.
    \end{align}
\end{assu}

\begin{remark}
If the given set $\mathbb X_0$ does not satisfy \eqref{X0inv}, a suitable outer approximation can be used to ensure invariance.
\end{remark}

\subsection{Modification of the linear regression}\label{augmentedeq}
While an update law for learning $\theta$ can be derived using \eqref{yhat}, the regressors are further modified to improve convergence of the observer by including an additional $qn+mn+n-1$ independent error signals \cite{kreisselmeier1977adaptive,lion1967rapid}. Additionally, update laws for the uncertainty sets $\Psi$ and $\mathbb X_0$ are constructed using the modified regression. To this end, define the following matrices:
\begin{subequations}
\begin{align}
   & \mathcal{W}_t\triangleq \begin{bmatrix}
    \omega_{t_{(0)}}^{\top}  &  \omega_{t_{(1)}}^{\top}  &   ... &  \omega_{t_{(qn+mn+n-1)}}^{\top} 
\end{bmatrix}^{\top}\nonumber\\
&\;\;\;\;\;\;\;\;\;\;\;\;\;\;\;\;\;\;\;\;\;\in\mathbb R^{q(qn+mn+n)\times(qn+mn+n)}, \label{Wtdef} \\
& \mathcal{Y}_t\triangleq \begin{bmatrix}
    y_{t_{(0)}}^{\top}  &  y_{t_{(1)}}^{\top}  &   ... &  y_{t_{(qn+mn+n-1)}}^{\top} 
\end{bmatrix}^{\top}  \in\mathbb R^{q(qn+mn+n)}, \label{Ytdef}
\end{align}
\end{subequations}
where the elements $\omega_{t_{(i)}}$ and $y_{t_{(i)}}$ are generated recursively using exponentially weighted filtering of current and past data, such that the following relation (similar to \eqref{ytrueobs}) holds.
\begin{align}
   y_{t_{(i)}}=\omega_{t_{(i)}}\theta+\eta_{t_{(i)}}, \label{yti_wti}
\end{align}
where $  \eta_{t_{(i)}}\in \mathcal N_{t_{(i)}}$ that are generated recursively from the disturbance bounds. The detailed recursive laws are given in Appendix \ref{app:ywNdefinitions}.

\subsection{Adaptation law}

At each $t$, the signals $\mathcal Y_t$ and $\mathcal W_t$ are available, which allows computation of a non-falsified set for the true quantity $\theta$ using \eqref{yti_wti}, \eqref{ch5ie:MatReg}, and consequently, updating of the uncertainty sets $\Psi$ and $\mathbb X_0$ using set membership identification \cite{boydset}. The non-falsified set is defined as
\begin{align}
    \Xi_{t}&\triangleq \Big{\{}(\bar p,\bar x_0) \;|\; y_{t_{(i)}}- \omega_{t_{(i)}}\begin{bmatrix}
       { \bar p}^\top {{\bar x}_0}^\top
    \end{bmatrix}^\top \in\mathcal N_{t_{(i)}}\nonumber
    \\& \forall (\bar x_0, \bar p, i)\in\mathbb R^n \times \mathbb R^{qn+mn} \times \mathbb I_0^{qn+mn+n-1}\Big{\}},
\label{nonfset}
\end{align}
where $\Xi_t\in\mathbb R^{qn+mn}\times \mathbb R^n$. The sets $\Psi$ and $\mathbb X_0$ are updated as
\begin{align}
   & \Pi_{t+1}\times \mathbb X_{0_{t+1}}\triangleq \{\Pi_{t}\times \mathbb X_{0_{t}}\}\cap \Xi_{t+1} \label{pix1}\\
   & \Psi_{t+1}\triangleq \{ \mathcal Z_1^{-1}(\bar p,f,q)\;|\;\forall \bar p\in\Pi_{t+1}, f \text{ in \eqref{obsF2}}\}\label{psi1}
\end{align}
with the initial sets
\begin{align}
    \Psi_0\triangleq \Psi,\; \Pi_0\triangleq \Pi,\;\mathbb X_{0_0}\triangleq \mathbb X_0.
\end{align}
The adaptation of the sets $\mathbb X_{0_t}$ and $\Psi_t$ may also lead to a change in the number of their vertices, denoted by $L_{x_t}$ and $L_{p_t}$, respectively, with 
\begin{align}
    L_{x_0}\triangleq L_x,\;\;L_{p_0}\triangleq L_p.
\end{align}

The reduction in the uncertainty sets for $p$ and $x_0$ is subsequently exploited in the design of a point-based update law, in which the parameter estimate is projected onto the updated sets. The adaptive law for the point estimate is the following projection modified normalized gradient descent-based scheme.
\begin{subequations}\label{adaptivelaw}
\begin{align}
&{\bar \theta_{t+1}\triangleq\hat \theta_{t}+\frac{\kappa \mathcal W_{t+1}^\top \left(\mathcal Y_{t+1}-\mathcal W_{t+1} \hat \theta_{t} \right)}{1+\mathrm{trace}(\mathcal W_{t+1}^\top \mathcal W_{t+1})} \;\;\;\forall t\in\mathbb I_0^\infty, \label{gd1}} \\
   &  \hat \theta_{t+1}=\begin{cases}\bar \theta_{t+1},\;\;\text{if }\bar p_{t+1}\in\Pi_{t+1} \text{ and }\bar x_{0_{t+1}}\in\mathbb X_{0_{t+1}}\\
     \begin{bmatrix}
        \hat p_{t+1}\\
       \hat x_{0_{t+1}}
    \end{bmatrix},\;\;\text{otherwise,}\end{cases}\label{proj2}
\end{align}
\end{subequations}
where $\kappa\in(0,2)$,
\begin{subequations}\label{forprojection}
\begin{align}
&\bar p_{t+1}\triangleq \bar \theta_{t+1}(1:qn+mn,1),\\
&\bar x_{0_{t+1}}\triangleq\bar \theta_{t+1}(qn+mn+1:qn+mn+n,1),\\
& \hat p_{t+1}\triangleq \arg\underset{\zeta\in\Pi_{t+1}}{\min}\|\bar p_{t+1}-\zeta\|_2,\\
&  \hat x_{0_{t+1}}\triangleq \arg\underset{\zeta\in\mathbb X_{0_{t+1}}}{\min}\|\bar x_{0_{t+1}}-\zeta\|_2.
\end{align}
\end{subequations}
The estimates $\hat A_{t+1}$ and $\hat B_{t+1}$ are obtained as
\begin{subequations}\label{ABfrom_projection}
\begin{align}
&  \hat \psi_{t+1}= \mathcal Z_1^{-1}(\hat p_{t+1},f,q),\\
& \hat A_{t+1}=\hat \psi_{t+1}(:,1:n),\\
&  \hat B_{t+1}=\hat \psi_{t+1}(:,n+1:n+m).
\end{align}
\end{subequations}
Further, since $\mathbb X_{0_t}$ is updated, the definition \eqref{X0ttildeset0} is improved to
\begin{align}\label{X0ttildeset}
    \widetilde{\mathbb X}_{0_t}\triangleq  \{x\;|\;x+\hat x_{0_t}\in \mathbb X_{0_t}\}.
\end{align}
The steps for implementing the observer are given in Algorithm \ref{alg:obs}. 

\begin{lemma}\label{lemma:nonfset}
If the true parameters and initial condition of \eqref{sys1x} belong to the initial uncertainty sets, i.e., $p\in\Pi_0$ (or $\psi\in\Psi_0$) and $x_0\in\mathbb X_{0_0}$, then they also belong to the update sets obtained using \eqref{nonfset}-\eqref{psi1}, i.e.,
\begin{align*}
    p\in\Pi_t,\;\;\psi\in\Psi_t, \;\;x_0\in\mathbb X_{0_t}\;\;\;\forall t\in\mathbb I_0^\infty,
\end{align*} where the updated sets are nested, i.e.,
\begin{align*}
    \Pi_t\subseteq \Pi_{t-1},\;\;\Psi_t\subseteq\Psi_{t-1},\;\;\mathbb X_{0_t}\subseteq\mathbb X_{0_{t-1}}\;\;\;\forall t\in\mathbb I_1^\infty.
\end{align*}
Consequently, the sets $\Pi_t$, $\Psi_t$, $\mathbb X_{0_t}$ and $\widetilde{\mathbb X}_{0_t}$ are non-empty $\forall t\in\mathbb I_0^\infty$.
\end{lemma}

\begin{proof}
    Refer to Appendix \ref{app:nonfset}.
\end{proof}

For the convergence analysis, the minimal robust positive invariant (RPI) set generated as $\lim_{t\rightarrow \infty}\bigoplus_{i=0}^{t-1} F^t\mathbb D$ is required. To keep the computation tractable,\begin{align}\label{SetDRPI}
\mathbb D^{RPI} =\text{ Outer RPI approximation of } \lim_{t\rightarrow \infty}\bigoplus_{i=0}^{t-1} F^t\mathbb D
\end{align}
is used, which can be obtained using \cite{rakovic2005invariant}. The set $\mathbb D^{RPI}$ contains all possible values of the disturbance-related summation term in \eqref{xtrue}. Accordingly, the minimal RPI set containing $\eta_t$ is $C\mathbb D^{RPI}$.

\begin{theo}\label{theo:observer}
Suppose the input $u_t$ to the plant \eqref{sys1x} and the adaptive observer \eqref{adapobs1}, and the output $y_t$ in \eqref{sys1y} are uniformly bounded. Define an error $e_t\triangleq ({\mathcal Y_t- \mathcal W_t \hat \theta_{t-1}})/({1+\mathrm{trace}(\mathcal W^{\top}_{t} \mathcal W_{t})})$, where $\mathcal Y_t$, $\mathcal W_t$ are matrices formed with \eqref{Wtdef}-\eqref{yti_wti}, \eqref{ch5ie:MatReg} using the data available on system \eqref{sys1}. The update law in \eqref{adaptivelaw} with \eqref{forprojection}, \eqref{ABfrom_projection} guarantees
\begin{itemize}
\item[(1)] $e_t$, $ e_t\left(1+\mathrm{trace}(\mathcal W^{\top}_{t} \mathcal W_{t}) \right)^{\frac{1}{2}}$, $\hat\theta_t$ $\in\mathcal L_\infty$,
\item[(2)] $e_t$, $ e_t (1+\mathrm{trace}(\mathcal W^{\top}_{t} \mathcal W_{t}) )^{\frac{1}{2}}$, $\|\hat\theta_t-\hat\theta_{t-1}\|_2$ $\in \mathcal S({\bar \eta}^2)$,\footnote{A sequence vector $z_t$ is said to belong to $\mathcal{S}({\bar \eta}^2)$ if $\sum_{i=t}^{t+k}z^{\intercal}_iz_i\leq c_0{\bar \eta}^2 k+c_1$ $\forall \;t\in\mathbb{I}_1^\infty$, a given constant ${\bar \eta}^2$, and some $k\in\mathbb{I}_1^\infty$, where $c_0,\;c_1\geq0$ \cite[Theorem~4.11.2, footnote~6]{ioannou2006adaptive}.} where $\bar \eta$ is the upper bound of $\|\eta  (1+\mathrm{trace}(\mathcal W^{\top}_{t} \mathcal W_{t})^{-\frac{1}{2}}\|_2$, and $\eta\in C\mathbb D^{RPI}$, defined in \eqref{SetDRPI}.
        \item[(3)] the state estimation error $\tilde{x}_t =x_t-\hat x_t \in\mathcal{L}_\infty$.
    \end{itemize}
\end{theo}

\begin{proof}
The proof directly follows from \cite[Lemma~1]{anchita2025}. Note that $\mathcal W_t$ is formed using $C$, $F$ which are constants, and $M_t$ whose dynamics is BIBO stable; this implies $\mathcal W_t$ has a finite upper bound. Similarly, since $\sigma\in(0,1)$ (refer to Appendix \ref{app:ywNdefinitions}), the worst- case set for $\eta_{t_{(i)}}$ in \eqref{yti_wti} is $C\mathbb D^{RPI}$, implying that $\bar\eta$ is a finite known constant.
\end{proof}

\section{Output Feedback Adaptive Tube MPC}\label{sec:COCPcomponents}
Building on the adaptive observer framework developed earlier and motivated by \cite{anchita2025, anch}, the COCP is reformulated in terms of the quantities available at each time step: the state estimate (observer state) $\hat x_t$, its initial condition $\hat x_{0_t}$, the parameter point estimates $\hat A_t$, $\hat B_t$, and the uncertainty sets $\Psi_t$, $\mathbb X_{0_t}$. All COCP components are updated as these estimates evolve.  

Analysis proceeds in two stages. First, the point estimate is held fixed, reducing the problem to a standard observer-based tube MPC formulation. Recursive feasibility and stability then follow from classical arguments \cite{kogel2017robust,mayne2006robust,mayne2009robust,rakovic2012homothetic}. Crucially, contraction of uncertainty sets $\Psi_t$ and $\mathbb X_{0_t}$ only relaxes the constraint tightening and enlarges the feasible region; it therefore cannot invalidate a previously feasible solution. Stability and feasibility guarantees are thus preserved even as the set shrinks.

In the fully adaptive setting, however, the point estimate and consequently, the prediction model and the feasible region, may change at every sampling instant. The COCP is therefore time-varying, and the standard argument that relies on shifting the previous optimal solution forward by one step no longer applies. Recursive feasibility and stability must be established separately through a suitable backup solution.



The COCP components are first constructed for a fixed-point estimate and then extended to the fully adaptive case. Two sources of error arise in this setting: state estimation error and prediction error, as discussed in the following subsections.

\subsection{Constraint tightening for state estimation error}\label{sec:xtilde}

Feasibility of the COCP ensures constraint satisfaction for the observer state but does not directly guarantee constraint satisfaction for the true plant state. This gap is addressed through constraint tightening based on the state estimation error dynamics. From \eqref{xtrue} and \eqref{xhat},
\begin{align}
    \tilde x_t=M_t\tilde p_t+F^t\tilde x_{0_{t}}+\sum_{k=0}^{t-1}F^{t-1-k}d_k,
\end{align}
where $\tilde p_t\triangleq p-\hat p_t$ and $\tilde x_{0_{t}}\triangleq x_0-\hat x_{0_{t}}$, are the parameter estimation error and initial state estimation error, respectively.

The COCP at time $t$ uses the point estimates available at that time, which are fixed over the prediction horizon. The predicted observer state and corresponding estimation error are denoted by $\hat x_{t,i}$ and $\tilde x_{t,i}$, respectively, where $\hat x_{t,i}$ represents the $i$-step-ahead predicted observer state and $\tilde x_{t,i}$ the corresponding estimation error\footnote{The notation $(\cdot)_{t,i}$ corresponds to the case of fixed point estimates, and is used to avoid ambiguity with the actual signal $(\cdot)_{t+i}$.}. With fixed point estimates, 
\begin{align}
    \hat x_{t,i}=M_{t+i}\hat p_t+F^{t+i}\hat x_{0_t}\;\;\;(\text{from \eqref{xhat}}),
\end{align}
which when subtracted from \eqref{xtrue} at time $t+1$ yields
\begin{align}
   &\; \tilde x_{t,1}= x_{t+1}-\hat x_{t,1}
    =\left(FM_{t}+\begin{bmatrix}
       Y_t & U_t
   \end{bmatrix}\right)\tilde p_t \nonumber\\
  & +F^{t+1}\tilde x_{0_t} +\sum_{k=0}^{t-1}F^{t-k}d_k+d_t \nonumber\\
 \Rightarrow &\;\tilde x_{t,i+1}=F\tilde x_{t,i}+\begin{bmatrix}
Y_{t+i} & U_{t+i}
\end{bmatrix}\tilde p_t +d_{t+i};\;\tilde x_{t,0}=\tilde x_t\nonumber\\
&\hspace{3cm}\forall t\in\mathbb I_0^\infty,\;i\in\mathbb I_0^{N-1}.
\label{state_est_err_dyn}
\end{align}
Propagating the dynamics in \eqref{state_est_err_dyn} yields the sets $\widetilde{\mathbb X}_{t,i}\ni \tilde x_{t,i}$:
\begin{align}
&\widetilde{\mathbb X}_{t,i}= F^{t+i} \widetilde{\mathbb X}_{0_t} \oplus \bigoplus_{k=0}^{t+i-1}F^{t+i-1-k}\left( \mathfrak{D}_{yu_{t}} \oplus \mathbb D \right)\label{TildeXti}\\
   &\mathfrak D_{yu_{t}}\triangleq \left\{ \begin{bmatrix}
       I_n\otimes y^\top & I_n\otimes u^\top
   \end{bmatrix}\tilde p \; \middle|\; y\in C\mathbb X,\right. \nonumber\\
   &\hspace{4cm}\left.\vphantom{\begin{bmatrix}
       I_n\otimes y^\top
   \end{bmatrix}}\;u\in\mathbb U,\;\tilde p \in \Pi_t\ominus \hat p_t\right\}.\label{Dyut}
\end{align}
Constraint sets for the state estimates are given by
\begin{align}
    \hat x_{t,i}\in \widehat{\mathbb X}_{t,i}\triangleq \mathbb X\ominus \widetilde{\mathbb X}_{t,i}.\label{HatXti}
\end{align}
For the COCP to be feasible, the tightened sets $\widehat {\mathbb X}_{t,i}$ must be non-empty. To characterize a conservative worst-case bound on the estimation error sets $\widetilde{\mathbb X}_{t,i}$ over the prediction horizon, define
\begin{align}
    \bar{\mathcal X}_t\triangleq F^t \widetilde{\mathbb X}_{0_t}\oplus \mathfrak D^{RPI}_{yu_t}\oplus \mathbb D^{RPI},
\end{align}
where $\bar{\mathcal X}_t\supseteq \widetilde{\mathbb X}_{t,i}$, and
\begin{align}\label{SetDyuRPI}
    \mathfrak D_{yu_t}^{RPI}=\text{Outer RPI approximation of}\nonumber\\
\lim_{t\rightarrow\infty}\bigoplus_{k=0}^{t-1} F^{t-1-k}\mathfrak {D}_{yu_t}
\end{align}
is obtained using \cite{rakovic2005invariant}. Consequently, if $\mathbb X\ominus \bar{\mathcal X}_t$ is non-empty, then the tightened sets $\widehat{\mathbb X}_{t,i}$ also remain non-empty. Accordingly, the following standard assumption is imposed \cite{kouvaritakis2016model,Rawlings2017MPC,CHISCI20011019,mayne2009robust,kogel2017robust,mayne2006robust,langson2004robust}.
\begin{assu}\label{assum:feasibility}
    The set $\mathbb X\ominus \bar{\mathcal X}_0$ is non-empty.
\end{assu}
This assumption ensures that at the initial time, even after subtracting the worst-case uncertainty effects, the set $\mathbb X$ retains admissible states, i.e., the tightened set remains non-empty to allow running a COCP.

\begin{remark}
If Assumption \ref{assum:feasibility} holds and the point estimate $\hat\theta_t$ remains fixed, then the sets $\mathbb X\ominus \bar{\mathcal X}_t$ remain non-empty and expand monotonically with time, i.e., $\mathbb X\ominus \bar{\mathcal X}_t \supseteq \mathbb X\ominus \bar{\mathcal X}_{t-1}$ for all $t$ (by Assumption \ref{Assum:1}). Further, for implementation purposes, Assumption \ref{assum:feasibility} is verified for a fixed point estimate rather than over all elements of $\Psi_0$ which reduces conservatism. Also, $\bar{\mathcal X}_0$ is a conservative bound; since the first term in $\widetilde{\mathbb X}_{t,i}$ decreases with time and the summation terms converge, tighter bounds on $\widetilde{\mathbb X}_{t,i}$ can be obtained at each step.
\end{remark}

\subsection{Constraint tightening for prediction error}\label{sec:predError}

The prediction model in the COCP propagates the observer state using \eqref{xhat} with fixed point estimates, given by 
\begin{align}
   \hat x_{t,i+1}=&M_{t+i+1}\hat p_t+F^{t+i+1}\hat x_{0_t}\nonumber\\
   =&\left(FM_{t+i}+ \begin{bmatrix}
       Y_{t+i}&U_{t+i}
   \end{bmatrix}\right)\hat p_t+ F^{t+i+1}\hat x_{0_t}\nonumber\\
   =& F \hat x_{t,i}+  \begin{bmatrix}
       Y_{t+i}&U_{t+i}
   \end{bmatrix}\hat p_t \nonumber\\
   =& F \hat x_{t,i} + (\hat A_t -F) x_{t+i}+\hat B_t u_{t+i},\nonumber
\end{align}
which depends on future values of the true state $x_{t+i}$ that are not available. To eliminate this dependence, the dynamics is rewritten in terms of available quantities as
\begin{align}
    \label{adobsdyn}
  \hat x_{t,i+1} = \hat A_t \hat x_{t,i}+\hat B_t u_{t+i} + \underbrace{(\hat A_t - F)\tilde x_{t,i}}_{\varepsilon_{t,i}},
\end{align}
which is used for state prediction within the COCP. The term $\varepsilon_{t,i}$ represents the mismatch between the prediction model and the observer dynamics, and is referred to as the \emph{prediction error}, where
\begin{align}
    \varepsilon_{t,i}\in\mathcal E_{t,i}\triangleq (\hat A_t-F) \widetilde{\mathbb X}_{t,i}.\label{mathcalEset}
\end{align}
This error is treated as a lumped disturbance within a robust homothetic tube-based MPC framework.


   Unlike \cite{anch,anchita2025,dhar2021indirect}, where the prediction error bounds grow with the horizon length $N$ to account for the unavailability of future parameter estimates, in the present formulation, the point estimates are held fixed over the prediction horizon. As a result, $\mathcal E_{t,i}$ is independent of $N$, which reduces conservatism and removes the restriction on the horizon length.

\subsection{Terminal set construction}\label{sec:TS}

The terminal set is chosen as the maximal admissible RPI set for the prediction dynamics in \eqref{adobsdyn} with respect to the prediction error $\mathcal E_{t,N-1}$ for $t\in\mathbb I_0^\infty$. This maximizes the terminal region and reduces the control effort required to steer the observer state into it, potentially allowing shorter horizons and fewer decision variables. The worst-case prediction error at the terminal step is
\begin{align}\label{Bar_E_allN}
    \bar{ \mathcal E}_t\triangleq \left(\hat A_t - F\right)\left(F^{t+N-1}\widetilde{\mathbb X}_{0_t} \oplus \mathfrak{D}_{yu_t}^{RPI} \oplus \mathbb D^{RPI} \right),
\end{align}
and the tightened terminal constraint set is
\begin{align}
   {\mathcal X}_t\triangleq {\mathbb X \ominus \left( F^{t+N} \widetilde{\mathbb X}_{0_t} \oplus \mathfrak D_{yu_t}^{RPI}\oplus \mathbb D^{RPI} \right)}.
\end{align}

To ensure feasibility and stability, the terminal set $\widehat{\mathbb X}_{TS_t}$ at initialization ($t=0$) is constructed under the following standard assumption \cite{kouvaritakis2016model,Rawlings2017MPC,CHISCI20011019,langson2004robust,fleming2015,buj2021,jafari2017adaptive,kohlerrampc, kogelblending, tranosrusso,wang2025tube,lorenzen2019robust,dhar2021indirect,anchita2025,anch,kogel2017robust,mayne2006robust,mayne2009robust,rakovic2012homothetic,lu2023robust}.
\begin{assu}\label{assu:XTS0}
For any $\begin{bmatrix}\hat A_0 & \hat B_0\end{bmatrix}\in\Psi_0$, and the pair $(P_0,\;K_0)$ computed using \eqref{originaldare} with $Q,\;R\succ 0$, there exists a non-empty terminal set $\widehat{\mathbb X}_{TS_0}$ containing the origin in its interior such that
\begin{gather}\label{ch7:marpi1}
  \begin{aligned}
   &\widehat{\mathbb X}_{TS_0}\subseteq \mathcal X_0,\;\;K_0 \widehat{\mathbb X}_{TS_0}\subseteq \mathbb U,\\
   &(\hat A_0 +\hat B_0 K_0) \widehat{\mathbb X}_{TS_0} \oplus\bar{\mathcal E_0}\subseteq \widehat{\mathbb X}_{TS_0}.
  \end{aligned}
\end{gather}
\end{assu}

When the point estimate $\hat \theta_t$ is updated, the terminal ingredients $P_t$ and $K_t$ and the terminal set must be recomputed. Such updates may compromise stability and feasibility\footnote{Recursive feasibility may be affected for small prediction horizons or empty tightened constraints; this is discussed later in Sec.~\ref{sec:recfeas}.}. While Assumption~\ref{assu:XTS0} can be verified offline, additional conditions are required to ensure safe online switching to new point estimates. To this end, the following criterion is checked $\forall t\in\mathbb I_1^\infty$.

\begin{crit}\label{Crdare2}
Given $\hat A_t$, $\hat B_t$, $(P_{t-1}, K_{t-1})$, $Q,\;R\succ 0$, and sets $\mathcal X_t$, $\mathbb U$, $\bar {\mathcal E}_t$, select $(P_t,K_t)$ with $P_t\succ 0$ such that
\begin{subequations}
\label{subeqnTS}
\begin{align}
\text{(a)}\;\;& P_t-(\hat A_t+\hat B_t K_t)^{\top} P_t (\hat A_t+\hat B_t K_t)\nonumber\\
&\quad \quad -Q-K_t^{\top} R K_t \succeq 0, \label{ch7:heura}\\
\text{(b)}\;\;& P_{t-1}-(\hat A_t+\hat B_t K_t)^{\top} P_t (\hat A_t+\hat B_t K_t) \nonumber\\
&\quad \quad -Q-K_{t-1}^{\top} R K_{t-1} \succeq 0, \label{ch7:dare2cri}\\
\text{(c)}\;\;&\left.\begin{aligned}
    &\widehat{\mathbb X}_{TS_t}\subseteq \mathcal X_t,\;\;K_t \widehat{\mathbb X}_{TS_t}\subseteq\mathbb U,\\
&(\hat A_t+\hat B_t K_t) \widehat{\mathbb X}_{TS_t}\oplus \bar{\mathcal E}_t\subseteq \widehat{\mathbb X}_{TS_t}
\end{aligned} \right\}\label{2ndterminala}
\end{align}
\end{subequations}
where $\widehat{\mathbb X}_{TS_t}$ is non-empty and contains the origin in its interior.
\end{crit}

Conditions \eqref{ch7:heura} and \eqref{2ndterminala} are standard terminal set requirements and can be verified offline. Condition~\eqref{ch7:dare2cri} is the key adaptive ingredient---it enforces compatibility between consecutive point estimate updates by ensuring the Lyapunov function remains non-increasing along closed-loop trajectories despite changes in the prediction model, thereby maintaining stability.

Condition~\eqref{ch7:dare2cri} is less restrictive than the quadratic stabilizability condition \cite{khargonekar1990robust}, which requires a common $(P,K)$ for all $\begin{bmatrix} \hat A & \hat B \end{bmatrix}\in \Psi_0$; here, only consistency between successive estimates is enforced, avoiding the need for a single Lyapunov function over the entire uncertainty set.

If Criterion~\ref{Crdare2} cannot be satisfied, the point estimate is held fixed while the uncertainty sets are updated. This preserves both feasibility and stability, while still allowing reduction of uncertainty as new data becomes available.
\subsection{Tube parameterization}

At each time $t$, the COCP is solved using a homothetic tube-based approach \cite{rakovic2012homothetic,langson2004robust}. Two tubes are constructed: one for the observer state and the other for the control input. The geometry of the state tube-sections is defined by a polytope $\mathbb G_t$, designed based on the following standard assumption \cite{rakovic2012homothetic,lorenzen2019robust,anchita2025,anch,dhar2021indirect}.

\begin{assu}\label{assum:Gt}
The set $\mathbb G_t\triangleq \mathrm{\mathbf{co}}\left(\left\{ \mathfrak g_t^{[j]}\right\}_{j=1:H_t}\right)\subseteq \mathbb R^n$, with $H_t>0$ known vertices, is a polytope containing the origin in its interior and satisfying
\begin{subequations}
    \begin{align}
   & (\hat A_t + \hat B_t K_t)\mathbb G_t \oplus \hat{\mathcal E}_t  \subseteq \mathbb G_t,\\
& \text{where } \;  \hat{\mathcal E}_t\triangleq (\hat A_t-F)\bar{\mathcal X}_t.\label{Hat_E_allN}
\end{align}
\end{subequations}

\end{assu}  

For a given $t$ and prediction horizon $N$, $F^t \widetilde{\mathbb X}_{0_t} \supseteq F^{t+i}\widetilde{\mathbb X}_{0_t}$ for all $i\in\mathbb I_0^{N}$ (by Assumption~\ref{Assum:1}), while $\mathfrak D_{yu_t}^{RPI}\oplus \mathbb D^{RPI}$ captures all possible values of the summation terms in \eqref{TildeXti}. These imply that $\mathbb G_t$ is RPI for the prediction dynamics in \eqref{adobsdyn} with respect to all admissible values of prediction error $\varepsilon_{t,i}$. In practice, $\mathbb G_t$ is chosen as the minimal RPI set, computed using \cite{rakovic2005invariant}, to capture the worst-case prediction error. 

The observer state and control input tubes are defined as
\begin{subequations}\label{tubes_sc}
\begin{align}
& \mathcal T^{\hat x}_t\triangleq \left\{ \mathbb T^{\hat x}_{i|t} \right\}_{i=0:N},\quad \mathcal T^{u}_t\triangleq \left\{ \mathbb T^u_{i|t} \right\}_{i=0:N-1},\\
& \mathbb T^{\hat x}_{i|t} = \mathrm{\mathbf{co}}\left( \left\{\mathfrak s_{i|t}^{[j]} \right\}_{j=1:H_t}\right) \triangleq \alpha_{i|t}\oplus \beta_{i|t}\mathbb G_t,\\
& \qquad\Rightarrow \mathfrak s_{i|t}^{[j]}\triangleq \alpha_{i|t}+\beta_{i|t}\mathfrak g_t^{[j]},\\
& \mathbb T^{u}_{i|t} = \left\{\mathfrak u_{i|t}^{[j]} \right\}_{j=1:H_t}.
\end{align}
\end{subequations}
For any point $s\in\mathbb T^{\hat x}_{i|t}$, expressed as $s=\sum_{j=1}^{H_t} \tau^{[j]}_{(i,t)}\,\mathfrak{s}_{i|t}^{[j]}$, where each $\tau^{[j]}_{(i,t)}\in [0,1]$ and $\sum_{j=1}^{H_t} \tau_{(i,t)}^{[j]}=1$,
the control input is generated by the same convex combination of control vertices as follows:
\begin{align}\label{ConvU}
u= \mathcal Z_2(s,\mathbb T^{\hat x}_{i|t},\mathbb T^u_{i|t})
\triangleq\sum_{j=1}^{H_t}\tau^{[j]}_{(i,t)}\,\mathfrak{u}_{i|t}^{[j]}.
\end{align}

\subsection{Reformulated COCP}

Let $\mu_t\triangleq \left\{ \{(\alpha_{i|t}, \beta_{i|t})\}_{i=0:N},\; \left\{\mathfrak{u}^{[j]}_{i|t} \right\}_{i=0:N-1,\;j=1:{H_t}} \right\}$ denote the decision variable. The COCP is given by 
\begin{subequations}\label{MPC2}
\begin{align}
&\mathbb{P}_t:\;\;\min_{\mu_t}  \sum_{j=1}^{H_t} \left(\sum_{i=0}^{N-1} \left(\| {\mathfrak{s}}_{i|t}^{[j]}\|^2_Q + \| {\mathfrak{u}}_{i|t}^{[j]}\|^2_R \right)+\| {\mathfrak{s}}_{N|t}^{[j]}\|^2_{P_t} \right) \\
&\text{subject to } \; \eqref{subeqnTS}-\eqref{tubes_sc}, \nonumber\\
&\beta_{i|t}\geq 0\;\;\forall i\in\mathbb{I}_{1}^{N},  \label{cons1m}\\
& \hat x_t\in\mathbb T^{\hat x}_{0|t},  \label{cons2m} \\
& \mathbb{T}^{\hat x}_{i|t}\subseteq \widehat{\mathbb{X}}_{t,i},\;\;\mathbb{T}^u_{i|t}\subseteq\mathbb{U}\;\;\forall i\in\mathbb{I}_{0}^{N-1}, \label{cons3m} \\ 
& \mathbb{T}^{\hat x}_{N|t}\subseteq \widehat{\mathbb{X}}_{TS_t}, \label{cons4m}\\ 
&\hat{A}_t {\mathfrak{s}}_{i|t}^{[j]}+\hat{B}_t {\mathfrak{u}}_{i|t}^{[j]}\in \mathbb{T}^{\hat x}_{i+1|t}\ominus \mathcal E_{t,i} \;\;\forall (i,j)\in\mathbb{I}_{0}^{N-1}\times\mathbb{I}_1^{H_t}.  \label{cons6m}
\end{align}
\end{subequations}

The input applied to the plant \eqref{sys1x} and observer \eqref{adapobs1} is obtained from the optimal solution using \eqref{ConvU}, i.e.,
\[
u_t^*= \mathcal Z_2(\hat x_t,\mathbb T^{\hat x^*}_{0|t},\mathbb T^{u^*}_{0|t}).
\]

At time $t+1$, a new COCP is constructed using the updated sets $\Psi_{t+1}$, $\mathbb X_{0_{t+1}}$ and estimates $\hat A_{t+1}, \hat B_{t+1}$. This transition requires additional conditions to preserve stability and recursive feasibility. As subsequently shown in Theorem \ref{thm:stab}, stability is ensured if condition \eqref{ch7:dare2cri} holds. However, even when Criterion \ref{Crdare2} is satisfied, infeasibility of \eqref{MPC2} may arise due to a short prediction horizon or empty tightened sets $\widehat{\mathbb X}_{t+1,i}$. To address this, a fallback mechanism is employed by reverting to the previous estimates:
\begin{gather}\label{setup_backup}
\begin{aligned}
&\hat x_{0_{t+1}}\leftarrow\hat x_{0_{t}},\;\;\hat A_{t+1}\leftarrow\hat A_{t}, \;\;\hat B_{t+1}\leftarrow\hat B_{t},\;\;P_{t+1}\leftarrow P_{t},\\
&K_{t+1}\leftarrow K_{t},\;\;\Psi_{t+1}\leftarrow \mathrm{\mathbf{co}}\left(\Psi_{t+1}, \left\{ \begin{bmatrix}
        \hat A_{t} & \hat B_{t}
    \end{bmatrix}\right\}\right),\\
& \Pi_{t+1}\leftarrow \mathrm{\mathbf{co}}\left(\Pi_{t+1}, \left\{ \hat p_{t}\right\}\right),\;\;
\mathbb X_{0_{t+1}}\leftarrow \mathrm{\mathbf{co}}\left(\mathbb X_{0_{t+1}}, \left\{ \hat x_{0_{t}}\right\}\right).
\end{aligned}
\end{gather}
The observer state is then recomputed using \eqref{xhat}, and the COCP components are reconstructed as described in Sec.~\ref{sec:COCPcomponents}, after which the COCP is solved at time $t+1$.

Any update to $\Psi_t$, $\mathbb X_{0_t}$, or the estimates $\hat A_t$, $\hat B_t$, $\hat x_t$, $\hat x_{0_t}$ modifies the geometry of the sets involved in the COCP, including the polytope $\mathbb G_t$, yielding an \textit{adaptive tube}, enabling improved characterization of uncertainty in the propagation of both the observer state and the true plant state.
\begin{figure}[t!]
 \vspace{0.23cm} \centering
     {{\includegraphics[scale=0.13]{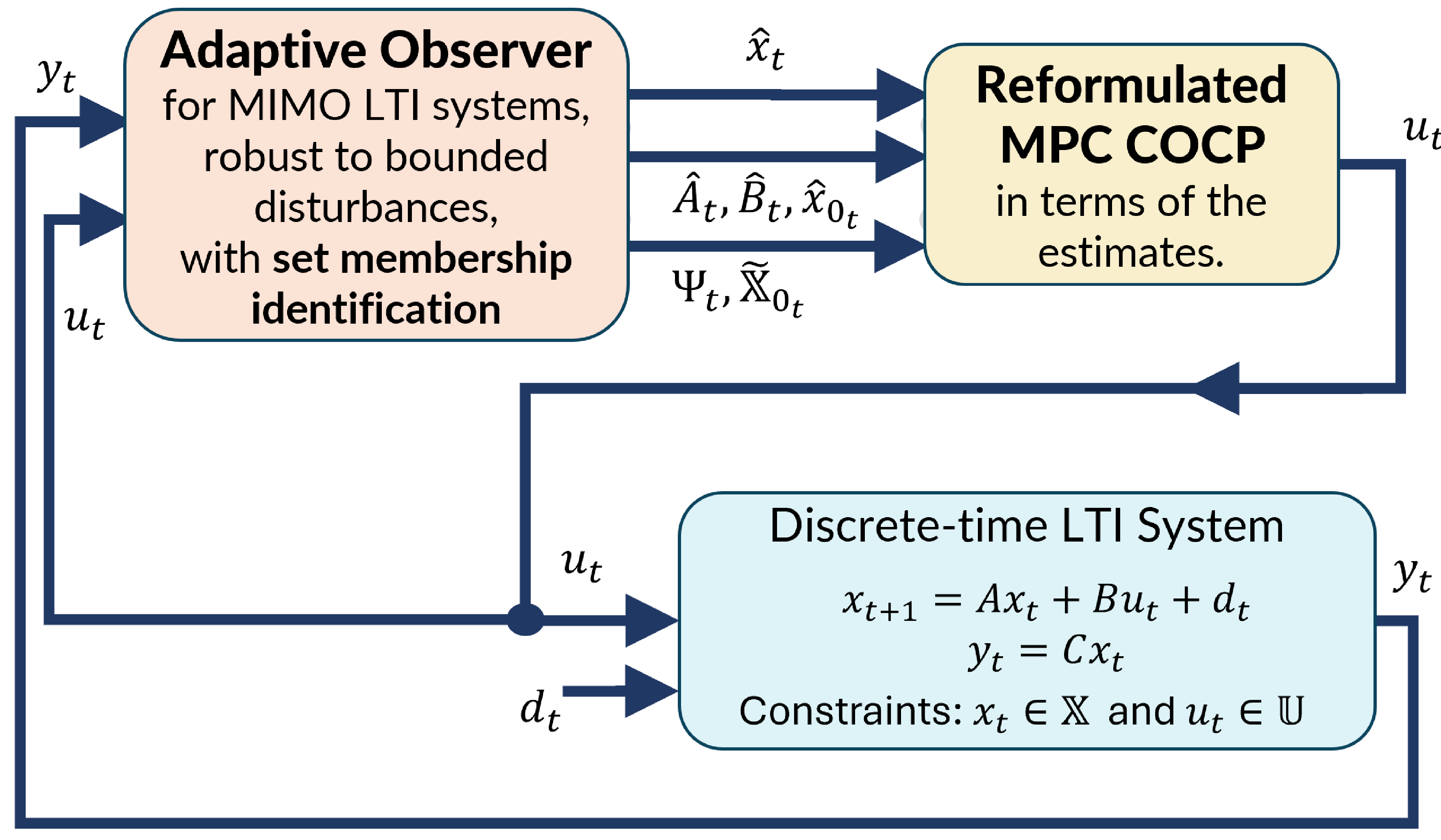}}}
     \caption{Schematic diagram for the proposed adaptive tube MPC. }
      \label{ch8:block_diag_RAOFMPC}
\end{figure} 
The implementation steps are summarized in Algorithm \ref{ch7:algoatube}, and the proposed scheme is illustrated in Fig.~\ref{ch8:block_diag_RAOFMPC}.

\begin{algorithm}[ht!]
\caption{Output Feedback MPC with Adaptive Tubes}\label{ch7:algoatube}
\begin{algorithmic}[1]
 \REQUIRE $\mathbb{X}$, $\mathbb{U}$, $\mathbb{D}$, $\psi^{[i]}$ $\forall i\in\mathbb{I}_1^{L_{p_0}}$ and $\hat \psi_0$ with $P_0$, $K_0$ satisfying \eqref{originaldare} for $\psi_0$, $F$, the $L_{x_0}$ vertices of $\mathbb{X}_0$ and $\hat x_{0_0}$ such that Assumption \ref{Assum:1} holds, $N$, $Q$, $R$, $\kappa$, $\sigma$.
 \ENSURE The optimal decision variable $\mu_t^*$ $\;\forall t\in\mathbb{I}_0^\infty$.\\
 \STATE Initialize the sets $\Psi_0=\Psi$, $\mathbb X_{0_0}=\mathbb X_0$, a constant $c_{backup}=0$ and time $t=0$.
\STATE Compute $\mathbb D^{RPI}$ using \eqref{SetDRPI} and \cite{rakovic2005invariant}.
\STATE Measure $y_0$.
\WHILE{$t\geq 0$} 
\STATE Compute the following using the references given in parentheses: $\Pi_t$ [\eqref{PiSetdefn}], $\hat \theta_t$ [\eqref{yhat}], $\widetilde{\mathbb X}_{0_t}$ [\eqref{X0ttildeset}], $\widetilde{\mathbb X}_{t,i}$ [\eqref{TildeXti}] and $\widehat {\mathbb X}_{t,i}$ [\eqref{HatXti}] $\forall i\in\mathbb I_0^N$, $\mathfrak D_{yu_t}$ [\eqref{Dyut}], $\mathfrak D_{yu_t}^{RPI}$ [\eqref{SetDyuRPI}], $\mathcal E_{t,i}$ [\eqref{mathcalEset}], $\bar{\mathcal E}_t$ [\eqref{Bar_E_allN}], $\hat{\mathcal E_t}$ [\eqref{Hat_E_allN}], $\mathbb G_t$ [Assumption \ref{assum:Gt} using \cite{rakovic2005invariant}].\label{off1step}\footnotemark
\IF{$t==0$}
 \STATE Compute $\widehat{\mathbb X}_{TS_t}$ satisfying \eqref{ch7:marpi1} using \cite{dey2024computation}.\label{off2step}
 \ELSE
 \STATE Check if $\exists$ $P_{t}$, $K_{t}$ and $\widehat{\mathbb X}_{TS_{t}}$ (computed using \cite{dey2024computation}) satisfying Criterion \ref{Crdare2}.
  \IF{Criterion \ref{Crdare2} does not hold}
     \STATE Set $c_{backup}\leftarrow1$ and $t\leftarrow t-1$, and go to Step \ref{ch7:stepjump}.
  \ENDIF
\ENDIF

\STATE Run the COCP $\mathbb P_t$ in \eqref{MPC2}.
\IF{\eqref{MPC2} is feasible, i.e., $\mu_t^*\neq\emptyset$}
 \STATE Set $c_{backup}\leftarrow0$.
 \ELSE
  \IF{$t==0$}
    \STATE Print `Initially infeasible setup', and exit the algorithm.
    \ELSE
    \STATE Set $c_{backup}\leftarrow1$ and $t\leftarrow t-1$.
    \ENDIF
\ENDIF
\IF{$c_{backup}==0$} 
  \STATE Apply $u_t^*$ computed using \eqref{ConvU} to the true system and measure $y_{t+1}$.
  \STATE Apply $u_t^*$, $y_t$ and $y_{t+1}$ to the adaptive observer.
  \STATE Obtain $\Psi_{t+1}$, $\mathbb X_{0_{t+1}}$, $\hat A_{t+1}$, $\hat B_{t+1}$, $\hat x_{0_{t+1}}$ and $\hat x_{t+1}$ using Algorithm \ref{alg:obs}. 
\ENDIF
\IF{$c_{backup}==1$} \label{ch7:stepjump}
\STATE Obtain the point estimates and uncertainty sets using \eqref{setup_backup}. 
\ENDIF
\STATE Update $t\leftarrow t+1$.
\ENDWHILE
\end{algorithmic}
\end{algorithm}

\begin{algorithm}[ht!]
\caption{Adaptive Observer}
\label{alg:obs}
\begin{algorithmic}[1]
\REQUIRE $t$, $F$, $C$, $\hat \psi_{t}$, $\hat x_{0_{t}}$, $u_t$, $y_t$, $y_{t+1}$, $q$, $n$, $m$.
\ENSURE $\Psi_{t+1}$, $\mathbb X_{0_{t+1}}$, $\hat A_{t+1}$, $\hat B_{t+1}$, $\hat x_{0_{t+1}}$, $\hat x_{t+1}$.\\
\STATE Compute $M_{t+1}$ using \eqref{FilterM}, $\omega_{t+1}$ and $\hat \theta_{t+1}$ using \eqref{yhat}, and $\mathcal Y_{t+1}$, $\mathcal W_{t+1}$, $\mathcal N_{{t+1}_{(i)}}$ $\forall i\in\mathbb I_0^{qn+mn+n-1}$ using \eqref{Wtdef}-\eqref{yti_wti}, \eqref{ch5ie:MatReg}.
\STATE Compute $\Xi_{t+1}$ using \eqref{nonfset}.
\STATE Update $\Pi_{t+1}$, $\mathbb X_{0_{t+1}}$ and $\Psi_{t+1}$ with \eqref{pix1} and \eqref{psi1}.
\STATE Update $\hat p_{t+1}$ and $\hat x_{0_{t+1}}$ using \eqref{adaptivelaw}, \eqref{forprojection}.
\STATE Obtain $\hat A_{t+1}$, $\hat B_{t+1}$ from \eqref{ABfrom_projection}, and $\hat x_{t+1}$ from \eqref{xhat}.
\end{algorithmic}
\end{algorithm}


\begin{remark}
    The number of vertices of $\Psi_t$ and $\mathbb X_{0_t}$ may grow with time. To preserve computational tractability, the adaptation of the sets can be halted once the number of vertices $L_{p_t}$ and $L_{x_t}$ reach prescribed upper bounds determined by the available computational resources. Alternatively, complexity-management techniques such as those proposed in \cite{lorenzen2019robust,CHISCI20011019,Veres01011999} may be employed.
\end{remark}

\begin{remark}
    Computing all the sets of Sec.~\ref{sec:COCPcomponents} and running the COCP again at $t\leftarrow t+1$ can be time-consuming. An alternative is to reuse the control input obtained from the previous step: $u_t=\mathcal Z_2 (\hat x_t, \mathbb T^{\hat x^*}_{1|t-1},\mathbb T^{u^*}_{1|t-1})$, along with the modified setup \eqref{setup_backup}.
\end{remark}

\begin{remark}
In computing $\widetilde{\mathbb X}_{t,i}$ using \eqref{TildeXti}, the set $\mathfrak D_{yu_t}$ consists of a component involving $\tilde p_t$ that is constant across the prediction horizon at any given time $t$. This component, along with the accumulated disturbance terms $\bigoplus_{k=0}^{t-1-k} F^{t-1-k}\mathbb D$ and $\mathfrak D_{yu_t}^{RPI}$, can be updated recursively across time steps, avoiding full recomputation at each $t$.
\end{remark}

\subsection{Recursive feasibility and stability analysis}\label{sec:recfeas}

\begin{theo}\label{thm:fixed_rec_feas}
  Consider the system \eqref{sys1} subject to the constraints \eqref{hc} for which Assumptions \ref{Assum:1}-\ref{assum:Gt} hold, and the adaptive observer that provides point estimates and updated sets using \eqref{adapobs1}-\eqref{ABfrom_projection}, \eqref{ch5ie:MatReg}. If the COCP $\mathbb P_t$ is feasible, then $\mathbb P_{t+k}$ is also feasible $\forall k\in\mathbb I_1^\infty$.
\end{theo}
\begin{proof}
Refer to Appendix \ref{app:thm:fixed_rec_feas}.
\end{proof}

\begin{theo}\label{thm:stab}
     Provided \footnotetext{For $t=0$, Steps \ref{off1step} and \ref{off2step} may be performed offline.}Assumptions \ref{Assum:1}-\ref{assum:Gt} hold for system \eqref{sys1} subject to the constraints \eqref{hc}, and the COCP $\mathbb P_t$ at $t=0$ is feasible, then, under Algorithms \ref{ch7:algoatube} and \ref{alg:obs},
     the following are guaranteed.
   \begin{enumerate}
       \item[(1)] The signals $x_t$, $\hat x_t$, $u_t$, $\hat\theta_t\in\mathcal L_\infty$.
       \item[(2)] The adaptive observer exhibits robust exponential stability, i.e., $\exists$ constants $c_1>0$, $c_2 \in (0,1)$ such that
\[
\left\|\hat x_t\right\| \leq c_1 c_2^t\left\| \hat x_0\right\| + c_3 \quad \forall t\in\mathbb I_0^\infty,
\]
for some bounded constant $c_3$ depending on the disturbance, state estimation error and prediction error bounds. Further, the state estimation error sets are uniformly bounded implying the actual plant also exhibits robust exponential stability.
   \end{enumerate}
       \end{theo}
\begin{proof}
    The proof for $x_t$, $\hat x_t$, $u_t$, $\hat \theta_t\in\mathcal L_\infty$ easily follows due to constraint tightening and recursive feasibility of the COCP, by Theorem \ref{thm:fixed_rec_feas}, and the observer properties guaranteed by Lemma \ref{lemma:nonfset} and Theorem \ref{theo:observer}. 
        
    For the proof of the second part, refer to Appendix \ref{app:thm_on_stab}. 
\end{proof}
\begin{remark}    
    The overall scheme is initialized with user-specified estimates, rather than through the observer in Algorithm \ref{alg:obs}. These are used in Algorithm \ref{ch7:algoatube} to solve the COCP and obtain the optimal decision variable $\mu_t^*$. The resulting input is then applied to the system to obtain the next output; the corresponding measurements are supplied to the observer to produce updated estimates via Algorithm \ref{alg:obs}. This ensures that the conditions of uniformly bounded input and output, as required for observer convergence using Theorem \ref{theo:observer}, are met from the outset.
\end{remark}
\begin{theo}\label{coro:luenberger}
When no adaptation is performed, i.e., the point estimate $\hat\theta_t$ and the uncertainty sets $\Pi_t$, $\Psi_t$ and $\mathbb X_{0_t}$ are held fixed, the resulting closed-loop scheme reduces to a Luenberger observer-based output feedback tube MPC where the parametric uncertainty is treated as an additive disturbance, without any loss of recursive feasibility, robustness and stability guarantees for the COCP $\mathbb P_t$. \end{theo}
\begin{proof}
    Refer to Appendix \ref{app:coro:luenberger}.
\end{proof}

\section{Simulation Results and Discussion}

A second-order LTI system\footnote{A second-order system is considered to facilitate visualization of the tube geometry.} is considered:
\begin{align*}
   & x_{t+1}=\begin{bmatrix}
        -1.2 & 1\\ 0.2 & 0
    \end{bmatrix}x_t+\begin{bmatrix}
        4\\-3.233
    \end{bmatrix}u_t+d_t,\\
   & y_t=\begin{bmatrix}
        1&0
    \end{bmatrix}x_t,
\end{align*}
subject to constraints $\|x_t\|_\infty \leq 40$, $\|u_t\|_\infty \leq 4$, disturbance bound $\|d_t\|_\infty \leq 0.1$, and parametric uncertainty set $\Psi_0=\mathrm{\mathbf{co}}\left(\left\{\psi_0^{[i]} \right\}_{i=1:3} \right)$\footnote{The vertices of $\Psi_0$ are $\psi_0^{[1]}=\begin{bmatrix}
      -1.1 & 1 & 4\\0.2 & 0 & -3.1  
    \end{bmatrix}$, $\psi_0^{[2]}=\begin{bmatrix}
      -1.2 & 1 & 4\\0.2 & 0 & -3  
    \end{bmatrix}$, and $\psi_0^{[3]}=\begin{bmatrix}
      -1.3 & 1 & 4\\0.2 & 0 & -3.6  
    \end{bmatrix}$.}. The initial set $\mathbb X_0$ is shown in grey in Fig.~\ref{fig:PsiX0} (b). The COCP is implemented with horizon $N=10$, weights $Q=I_2$, $R=0.1$, and observer parameters $\kappa=0.2$, and $\sigma=0.9$. The true and observer initial states are chosen as 
\begin{align*}
    x_0=\begin{bmatrix}
        12\\39
    \end{bmatrix},\quad
    \hat x_0=\begin{bmatrix}
        20\\31
    \end{bmatrix},\quad\text{with}\quad
    F=\begin{bmatrix}
        0.03 & 1\\ 0.01 & 0
    \end{bmatrix}.
\end{align*}

The trajectories of the true and observer states, and control input are shown in Figs.~\ref{figStateTube_x} and \ref{figinput}, respectively. The state tube-sections in Fig.~\ref{figStateTube_x}, given by $\hat x_t\oplus \mathbb T^{\hat x}_{0|t}$, enclose the true state trajectory at all times while satisfying the imposed state constraint. The input constraint is also satisfied.

\begin{figure}[t!]
    \vspace{0.23cm}\centering
\framebox{\parbox{3in}{\includegraphics[width=\linewidth]{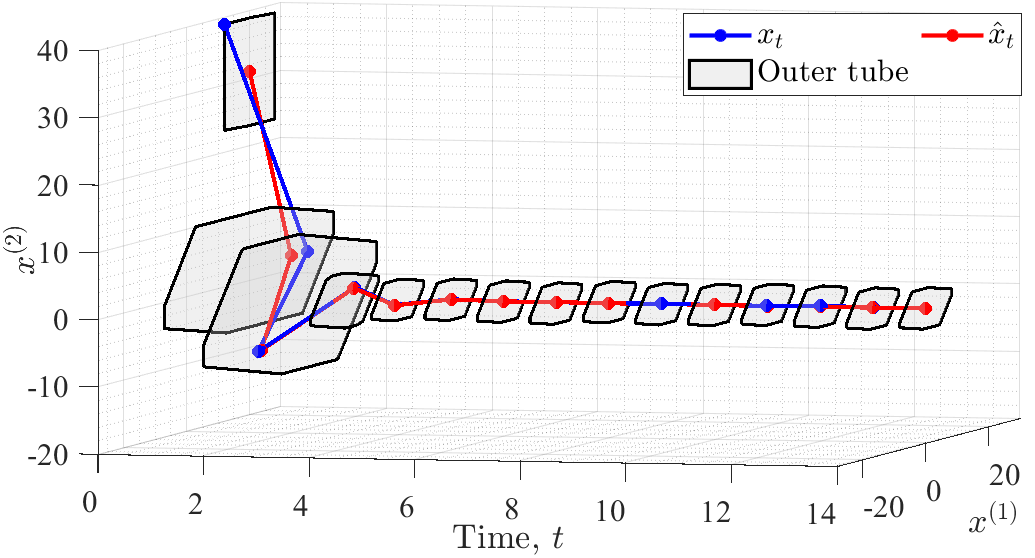}}}
    \caption{True state and state estimate trajectories, with the outer tube (in grey) guaranteed to contain the true state trajectory.}
    \label{figStateTube_x}
\end{figure}
\begin{figure}[t!]
    \vspace{0.23cm}\centering
\framebox{\parbox{3in}{\includegraphics[width=\linewidth]{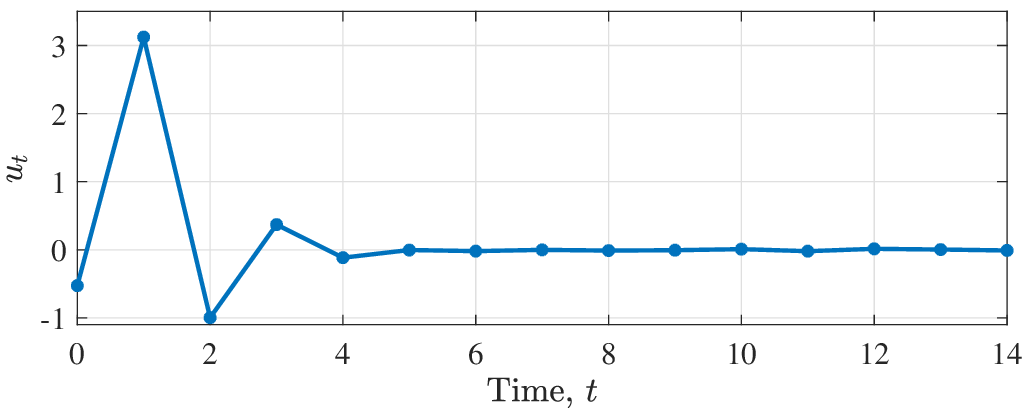}}}
    \caption{Control input.}
    \label{figinput}
\end{figure}

The adaptation of the tubes $\mathcal T^{\hat x}_t$ for the state estimate is illustrated in Fig.~\ref{fig:adaptivetubesall} for $t=0$ to $t=3$. Convergence of both the state and observer state trajectories to neighbourhoods of the origin is observed, along with decay of the input (see Figs.~\ref{figStateTube_x}, \ref{figinput}). The first two tube-sections returned by the COCP at previous time instant are shown in the background in the last three sub-figures of Fig. \ref{fig:adaptivetubesall}; the tube size visibly reduces with adaptation, giving a better characterization of the uncertainty in propagation of observer dynamics. The changes in the tube geometry are shown in Fig. \ref{fig:basicp}. 
\begin{figure}[t!]
    \vspace{0.23cm}\centering
\framebox{\parbox{3in}{\includegraphics[width=\linewidth]{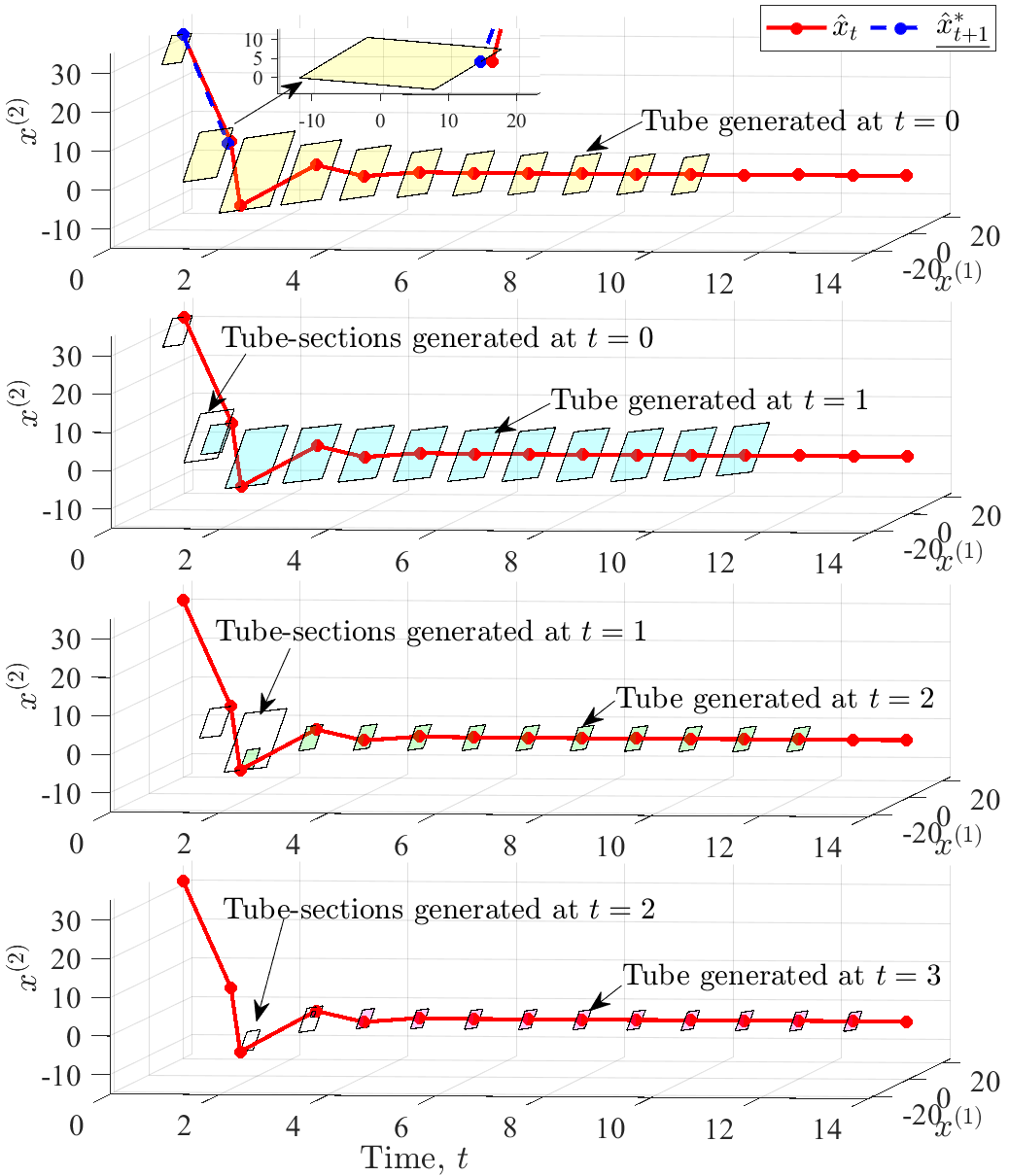}}}
    \caption{Tube for state estimate trajectory generated by the COCP at $t=0,\;1,\;2,\;3$, in yellow, cyan, green and pink, respectively.}
    \label{fig:adaptivetubesall}
\end{figure}
\begin{figure}[t!]
    \vspace{0.23cm}\centering
\framebox{\parbox{3in}{\includegraphics[width=\linewidth]{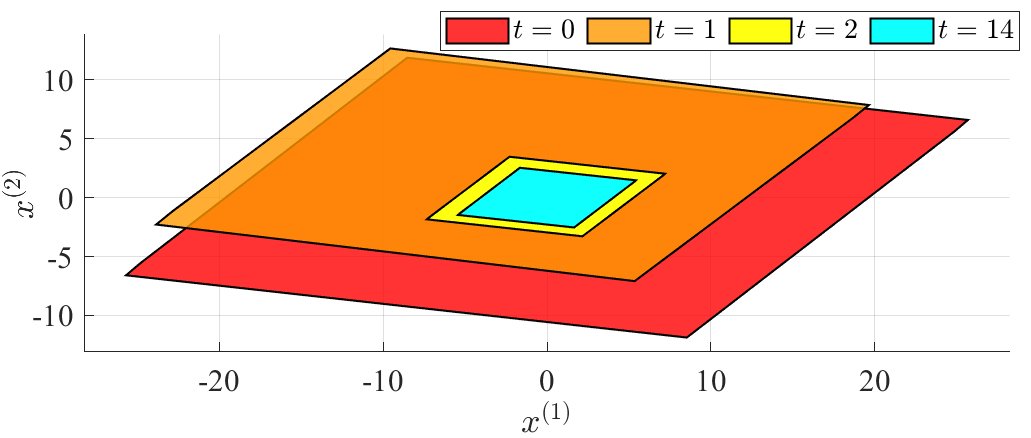}}}
    \caption{The polytopes forming the tube cross-sectional shape at time $t$.}
    \label{fig:basicp}
\end{figure}
\begin{figure}[t!]
    \vspace{0.23cm}\centering
\framebox{\parbox{3in}{\includegraphics[width=\linewidth]{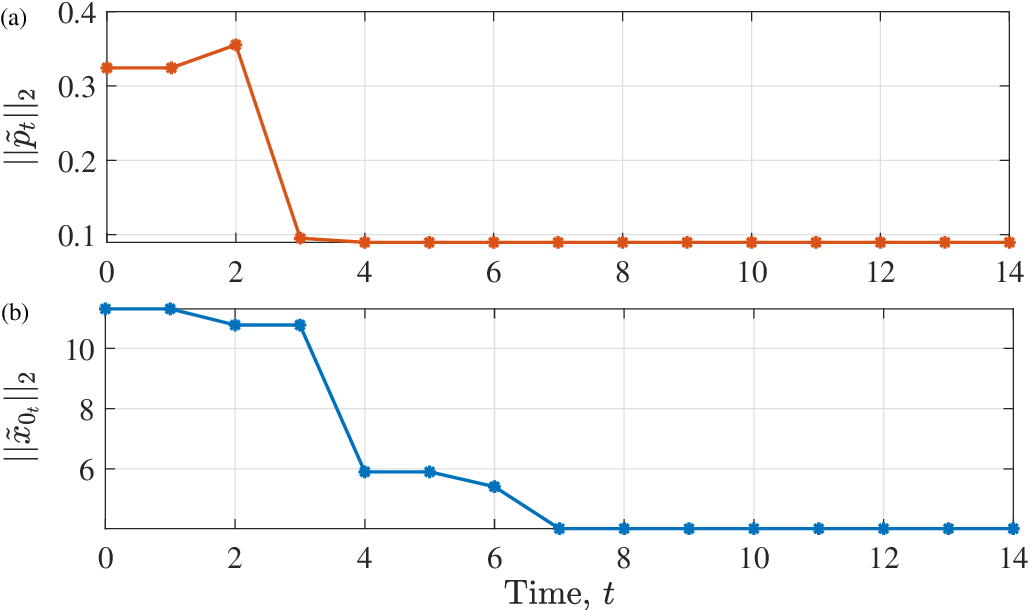}}}
    \caption{Norms of the parameter estimation error and the initial state estimation error.}
    \label{fig:theta}
\end{figure}

It is worth noting that the state estimate at $t=1$ is not contained in the tube-section $\mathbb T^{\hat x^*}_{1|0}$ (zoomed in first sub-figure of \ref{fig:adaptivetubesall}). This is expected; the yellow tube was generated under a fixed point estimate, whereas the state estimate is updated based on the improved uncertainty sets and point estimates of $\hat p_t$ and $\hat x_{0_t}$. To illustrate this effect, the first sub-figure additionally shows, using the blue dashed line $\underline{(\cdot)}$, the state estimate trajectory obtained when the point estimates are held fixed. As expected, this trajectory remains contained within the tube-section $\mathbb T^{\hat x^*}_{1|0}$.

\begin{figure}[t!]
    \vspace{0.23cm}\centering
\framebox{\parbox{3in}{\includegraphics[width=\linewidth]{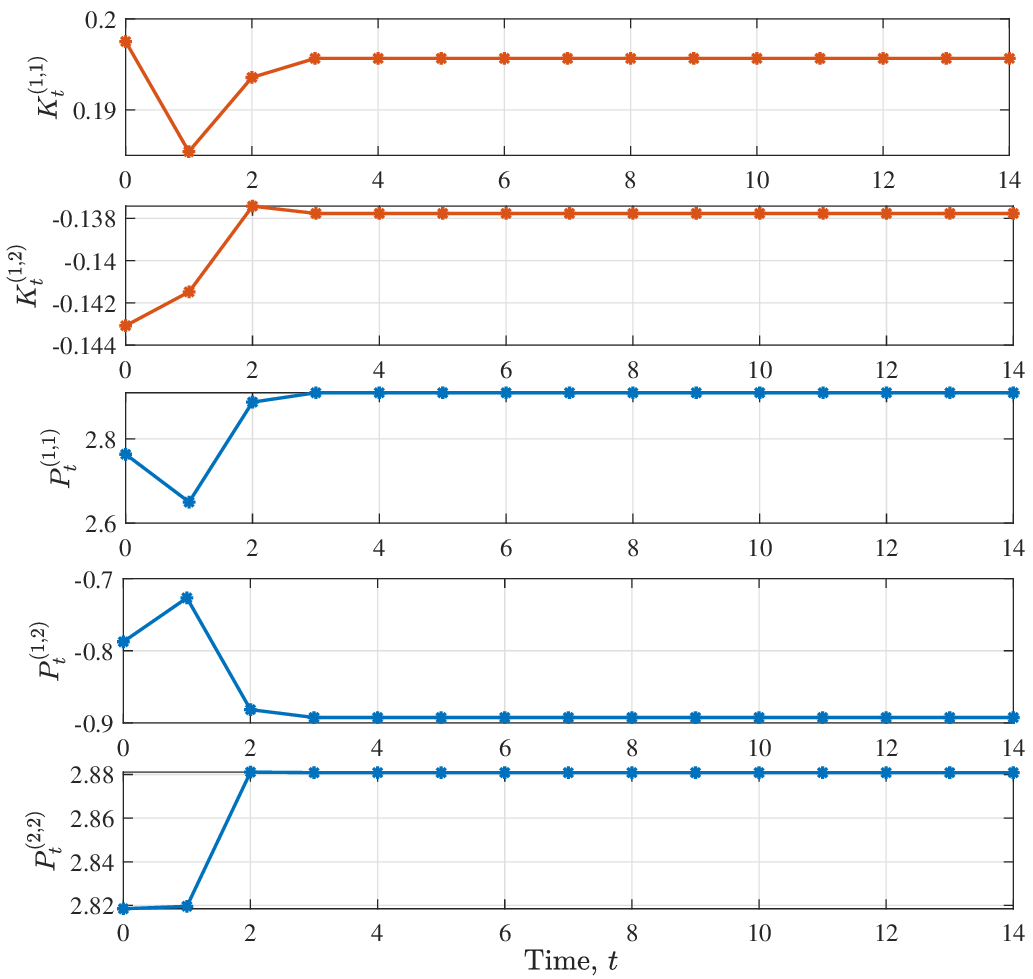}}}
    \caption{Variation of the elements of the stabilizing feedback gain $K_t$, and the terminal cost weight $P_t$.}
    \label{fig:PKt}
\end{figure}
\begin{figure}[t!]
    \vspace{0.23cm}\centering
\framebox{\parbox{3in}{\includegraphics[width=\linewidth]{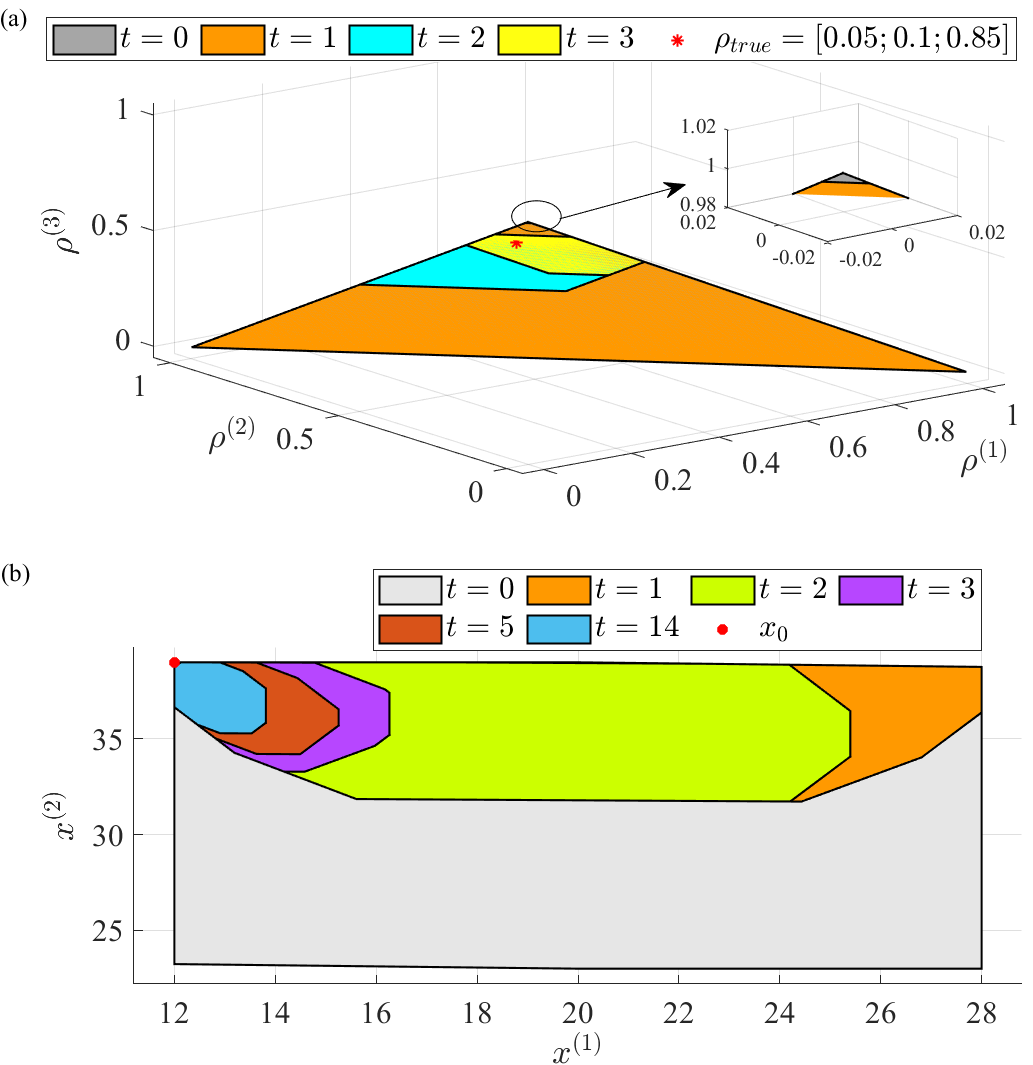}}}
    \caption{Reduction in size of uncertainty sets: (a) Barycentric representation of the sets $\Psi_t$ containing the true parameter combination $[0.05;0.1;0.85]$, using $\rho^{(i)}$. (b) Sets ${\mathbb X}_{0_t}$ containing the initial true state $x_0$.}
    \label{fig:PsiX0}
\end{figure}

The parameter and state estimation error norms $\| \tilde p_t\|_2$ and $\| \tilde x_{0_t}\|_2$, and the associated variations in $K_t$ and $P_t$ are displayed in Figs.~\ref{fig:theta} and \ref{fig:PKt}, respectively. The contraction of the uncertainty sets $\Psi_t$ and $\mathbb X_{0_t}$ with time are shown in Fig.~\ref{fig:PsiX0} (a) and (b), respectively, with Fig.~\ref{fig:PsiX0} (a) giving a barycentric representation using $\rho^{(i)}$ for $\psi_0^{[i]}$, for ease in visualization. Both the parameter estimation norm and the uncertainty set sizes $\Psi_t$ and $\mathbb X_{0_t}$ decrease with incoming data\footnote{An improved convergence may be achieved with a persistently exciting regressor. Incorporating excitation constraints on the COCP has been explored in \cite{marafioti2014persistently,lu2023robust}.}. Notably, the modification of the regressors in \eqref{ch5ie:MatReg} enables learning of the set $\mathbb X_{0_t}$ beyond $t=2$. 

A comparison with the non-adaptive case and \cite{kogel2017robust} is provided in Fig.~\ref{fig:comparexJ}, where the RMS value of $\|x_t\|_2$ and cumulative stage cost are shown. Improved performance is observed for the proposed adaptive scheme once sufficient learning has occurred. In the absence of adaptation, comparable performance with \cite{kogel2017robust} is observed, consistent with Theorem \ref{coro:luenberger}.
\begin{figure}[t]
    \vspace{0.23cm}\centering
\framebox{\parbox{3in}{\includegraphics[width=\linewidth]{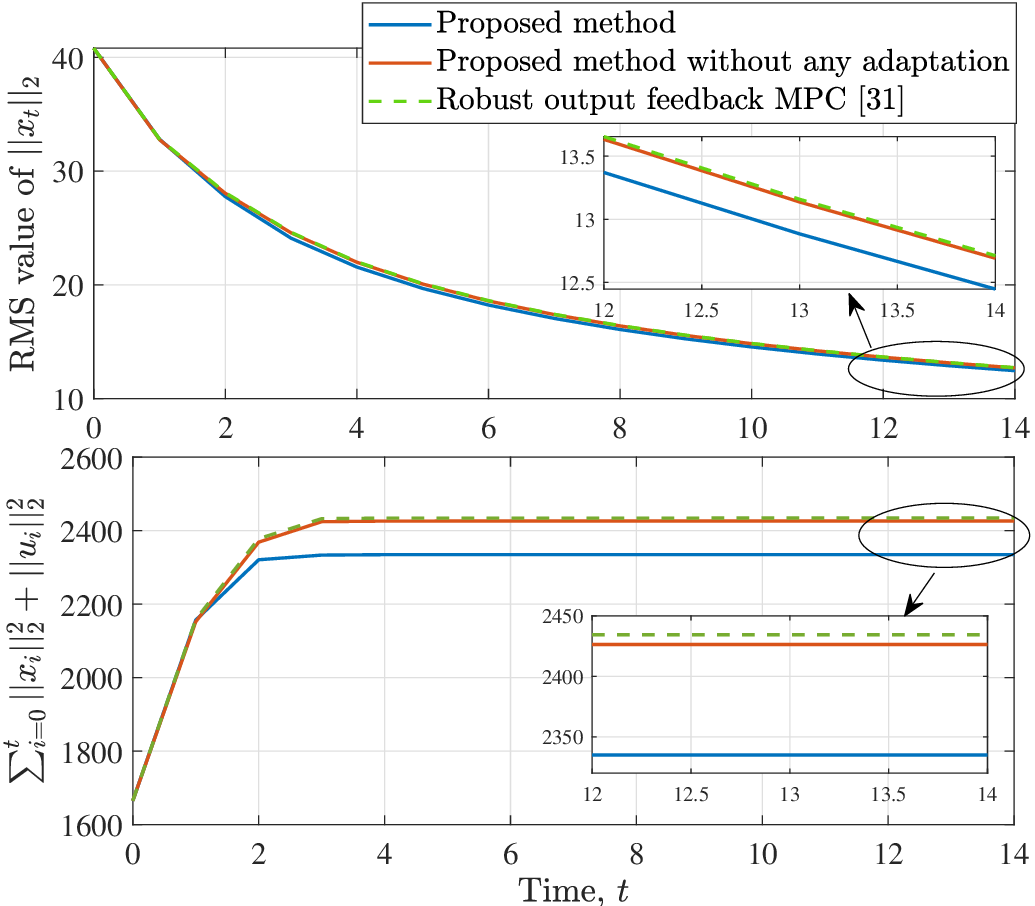}}}
    \caption{Comparison of RMS value of $\|x_t\|_2$ and the cumulative stage cost for the proposed method with and without any adaptation and \cite{kogel2017robust}.}
    \label{fig:comparexJ}
\end{figure}

\section{Conclusion}
An adaptive observer-based output feedback MPC framework with adaptive tubes has been developed for discrete-time LTI systems subject to parametric and additive uncertainties. Point and set estimates of the unknown parameters and initial state are obtained using an adaptive observer with joint set membership-based identification, enabling progressive refinement of uncertainty descriptions. The resulting scheme induces adaptation in the tube geometry and associated sets, thereby rendering the tube an adaptive nature. As uncertainty contracts with data, constraint tightening relaxes progressively, yielding improved closed-loop performance without sacrificing guarantees.

Recursive feasibility, robust exponential stability, and boundedness of all closed-loop signals are established for the resulting time-varying COCP. Crucially, these guarantees are obtained without requiring a common quadratically stabilizing linear feedback gain over the entire uncertainty set; a compatibility condition on successive terminal ingredients, together with a backup mechanism, suffices. In the limiting case of no adaptation, the framework recovers a standard Luenberger observer-based tube MPC. 

The result establishes a pathway from set membership identification to certified adaptive predictive control under output feedback. Future work will focus on reducing the computational complexity associated with online updates and extending the framework to a more general class of systems.

        
\appendices
\section{Bijective mapping between matrix and vector parameterizations}\label{app:z1z1inv}
The bijection $ \mathcal Z_1$ is defined as
\begin{align}
     \mathcal Z_1:&\; \mathbb R^{n\times (n+m)}\times \mathbb R^{qn}\times \mathbb{I}_1^\infty  \rightarrow  \mathbb{R}^{qn+mn} \nonumber\\
    &\left( \bar\psi,f,q  \right)\mapsto \begin{bmatrix}
       \mathrm{\mathbf{vec}} \left( \bar\psi(:,1:q)\right)-f \\ \mathrm{\mathbf{vec}} \left(\bar\psi(:,n+1:n+m) \right)
    \end{bmatrix},
\end{align}
i.e., $ \mathcal Z_1 \left(\begin{bmatrix}
    A&B
\end{bmatrix},f,q \right)= p$, and its inverse is
\begin{subequations}
\begin{align}
     \mathcal Z_1^{-1}&\;:\mathbb R^{qn+mn}\times \mathbb R^{qn} \times \mathbb I_1^\infty\rightarrow \mathbb R^{n\times (n+m)} \nonumber\\
    &(\bar p,f,q)\mapsto \left[\begin{array}{c|c|c}
   \bar{\mathcal{A}}\;&\begin{array}{cc}
         I_{n-q}\\0_{q\times (n-q) }
    \end{array} \;&\bar B
    \end{array}\right],\\
   \text{where } &\bar{ \mathcal A}\triangleq \mathrm{\mathbf{vec}}^{-1}(\bar p(1:qn,1)+f,n),\\
   & \bar B\triangleq  \mathrm{\mathbf{vec}}^{-1}(\bar p(qn+1:qn+mn,1),n).
\end{align} 
\end{subequations}

\section{Recursive construction of the augmented regression used by the adaptive observer}\label{app:ywNdefinitions}
The detailed recursive construction of the augmented regression signals and associated uncertainty sets used in Sec. \ref{augmentedeq} are given below.
\begin{subequations}\label{ch5ie:MatReg}
\begin{align}
&\omega_{0_{(i)}}\triangleq 0_{q\times(qn+mn+n)}\;\;\forall i\in\mathbb I_1^{qn+mn+n-1},\\
&\omega_{t_{(0)}}\triangleq \omega_t\;\;\forall t\in\mathbb I_0^\infty,\\
& \omega_{t_{(i)}}\triangleq \sigma \omega_{{t-1}_{(i)}}+(1-\sigma)\omega_{{t-1}_{(i-1)}} \nonumber\\
&\hspace{2.5cm}\forall (t,i)\in\mathbb {I}_1^\infty \times\mathbb I_1^{qn+mn+n-1}, \\
&y_{0_{(i)}}\triangleq 0_{q}\;\;\forall i\in\mathbb I_1^{qn+mn+n-1},\\
&y_{t_{(0)}}\triangleq y_t\;\;\forall t\in\mathbb I_0^\infty,\\
& y_{t_{(i)}}\triangleq \sigma y_{{t-1}_{(i)}}+(1-\sigma)y_{{t-1}_{(i-1)}} \nonumber\\
&\hspace{2.5cm}\forall (t,i)\in\mathbb {I}_1^\infty \times\mathbb I_1^{qn+mn+n-1},\end{align}\begin{align} 
&  \mathcal N_{0_{(i)}}\triangleq \{0\}\;\;\forall i\in\mathbb I_0^{qn+mn+n-1},\\
& \mathcal N_{t_{(0)}}\triangleq \bigoplus_{k=0}^{t-1} CF^{t-1-k}\mathbb D\;\;\forall t\in\mathbb I_1^\infty,\\
&  \mathcal N_{t_{(i)}}\triangleq \sigma \mathcal N_{t-1_{(i)}}\oplus (1-\sigma) \mathcal N_{t-1_{(i-1)}}\nonumber\\
&   \hspace{3cm}\forall (t,i)\in\mathbb {I}_1^\infty \times\mathbb I_1^{qn+mn+n-1} ,\label{calNSets} 
\end{align} 
with $\sigma\in (0,1)$.
\end{subequations}

\section{Proof of Lemma \ref{lemma:nonfset}}\label{app:nonfset}
The matrices $\mathcal Y_t$ and $\mathcal W_t$ are constructed using data measured from the system \eqref{sys1}. Consequently, the true quantities $p$ and $x_0$ satisfy
\begin{align*}
y_{t_{(i)}}-\omega_{t_{(i)}}\begin{bmatrix}
p^\top & x_0^\top
\end{bmatrix}^\top \in \mathcal N_{t_{(i)}}\; \forall (t,i)\in\mathbb I_0^\infty \times \mathbb I_0^{qn+mn+n-1}.
\end{align*}
From the definition of the non-falsified set in \eqref{nonfset}, it follows that
\begin{align}\label{Appeq1_new}
(p,x_0)\in \Xi_t \quad \forall t\in\mathbb I_1^\infty.
\end{align}

The uncertainty sets are updated through intersection as defined in \eqref{pix1}. Hence, the updated sets are nested and satisfy
\begin{align*}
\Pi_t \times \mathbb X_{0_t} \subseteq \Pi_{t-1} \times \mathbb X_{0_{t-1}} \quad \forall t\in\mathbb I_1^\infty.
\end{align*}
Since $(p,x_0)\in \Pi_0 \times \mathbb X_{0_0}$ and \eqref{Appeq1_new} holds, it follows by induction that
\begin{align*}
&(p,x_0)\in \Pi_t \times \mathbb X_{0_t} \quad \forall t\in\mathbb I_0^\infty \\
\Rightarrow &\;\psi\in \Psi_t, \quad
p\in \Pi_t,\quad x_0\in \mathbb X_{0_t},\quad  \forall t\in\mathbb I_0^\infty,\\
\Rightarrow &\;
\Psi_t \neq \emptyset, \quad \Pi_t \neq \emptyset,\quad  \mathbb X_{0_t} \neq \emptyset \quad \forall t\in\mathbb I_0^\infty.
\end{align*}

Finally, the update law in \eqref{adaptivelaw}, \eqref{forprojection} ensures that the estimate $\hat x_{0_t}$ is projected onto $\mathbb X_{0_t}$. As a result, the set $\widetilde{\mathbb X}_{0_t} = \mathbb X_{0_t} \ominus \hat x_{0_t}$ is non-empty for all $t\in\mathbb I_0^\infty$.

\section{Proof of Theorem \ref{thm:fixed_rec_feas}}\label{app:thm:fixed_rec_feas}
The proof proceeds by induction. Assume that the COCP \eqref{MPC2} is feasible at time $t$. It is shown that feasibility is preserved at time $t+1$.

Given a feasible solution at time $t$, the system evolves and the point estimates and uncertainty sets are updated. Criterion~\ref{Crdare2} is then evaluated at time $t+1$. If Criterion~\ref{Crdare2} holds, a new COCP is formulated with an updated terminal cost, terminal set, and feedback gain. Since the problem data change, feasibility of this newly formed problem cannot be guaranteed a priori. If the resulting problem $\mathbb P_{t+1}$ is feasible, the algorithm proceeds with its solution. Otherwise, or if Criterion~\ref{Crdare2} does not hold, the backup setup described in \eqref{setup_backup} is used.

Under the backup setup, two situations arise depending on the evolution of the estimates and uncertainty sets.

\begin{itemize}

\item[(i)] \emph{No adaptation of point estimate or uncertainty sets.}  
If there is no update in the point estimate $\hat{\theta}_t$ and the uncertainty sets $\Pi_t$, $\Psi_t$, and $\widehat{\mathbb X}_{0_t}$ after time $t$, the resulting COCP coincides with a standard robust tube-based MPC formulation. Recursive feasibility then follows directly from the classical robust MPC literature \cite{CHISCI20011019,rakovic2012homothetic,mayne2009robust,kogel2017robust}.

Let the feasible solution at $t+1$ for this completely non-adaptive scenario (i.e., no update of point estimates or uncertainty sets after $t$) be denoted using the notation $\underline{(\cdot)}$ as follows:




\begin{subequations}\label{nonada}
\begin{align}
      &  \underline{\mathbb T^{\hat x}_{i|t+1}}=\alpha^*_{i+1|t}\oplus \beta^*_{i+1|t}\underline{\mathbb G_{t+1}}\;\;\forall i\in\mathbb I_0^{N-1}\\
       & \underline{\mathbb T^u_{i|t+1}}=\{  \mathcal Z_2(\underline{\mathfrak s^{[j]}_{i+1|t}}, \mathbb T^{\hat x^*}_{i+1|t},\mathbb T^{u^*}_{i+1|t})\;|\;j\in\mathbb I_1^{H_{t+1}}  \})\nonumber\\
       &\hspace{5cm}\;\;\forall i\in\mathbb I_0^{N-2}\\
      &  \underline{\mathbb T^{\hat x}_{N|t+1}}=\underline{\alpha_{N|t+1}}\oplus \underline{\beta_{N|t+1}} \underline{\mathbb G_{t+1}}\\
      & \underline{\alpha_{N|t+1}}=(\hat A_t+\hat B_t K_t)\alpha_{N|t}^*\\
      & \underline{\beta_{N|t+1}}=\min_{\beta} \{ \beta \;|\; \beta \underline{\mathbb G_{t+1} }\supseteq (\hat A_t+\hat B_t K_t)\beta_{N|t}^* {\mathbb G_{t}}\nonumber\\
      &\hspace{5cm}\oplus \underline{\bar{\mathcal E}_{t+1}}\}\\ 
      &   \underline{\mathbb T^u_{N-1|t+1}}=K_t \underline{\mathbb T^{\hat x}_{N-1|t+1}}.
\end{align} \end{subequations}
The feasible solution \eqref{nonada} is adopted as a candidate solution in the next case to establish recursive feasibility.

\item[(ii)] \emph{Adaptation of uncertainty sets with fixed point estimates.} If the point estimates remain unchanged while the uncertainty sets are updated, the following relations hold: 
\begin{subequations}\label{eq:improved_sets}
\begin{align}
&\hat A_{t+1}=\hat A_t,\;\;\hat B_{t+1}=\hat B_t,\;\;\hat x_{0_{t+1}}=\hat x_{0_t},\\
&\Psi_{t+1}\subseteq\underline{\Psi_{t+1}}= \Psi_t,\;\;\Pi_{t+1}\subseteq \underline{\Pi_{t+1}}= \Pi_t,\\
& \mathbb X_{0_{t+1}}\subseteq \underline{\mathbb X_{0_{t+1}}}=\mathbb X_{0_{t}},\;\;\widetilde{\mathbb X}_{0_{t+1}}\subseteq \underline{\widetilde{\mathbb X}_{0_{t+1}}}=\widetilde{\mathbb X}_{0_{t}} ,\\
& \mathfrak D_{yu_{t+1}}\subseteq \underline{\mathfrak D_{yu_{t+1}}}=\mathfrak D_{yu_{t}} ,\;\;\mathfrak D_{yu_{t+1}}^{RPI}\subseteq  \underline{\mathfrak D_{yu_{t+1}}^{RPI}}=\mathfrak D_{yu_{t}}^{RPI},\\
& \widetilde{\mathbb X}_{t+1,i}\subseteq \underline{\widetilde{\mathbb X}_{t+1,i}}= \widetilde{\mathbb X}_{t,i+1},\;\;\widehat{\mathbb X}_{t+1,i}\supseteq \underline{\widehat{\mathbb X}_{t+1,i}}=\widehat{\mathbb X}_{t,i+1},\\
& \bar{\mathcal X}_{t+1} \subseteq \underline{\bar{\mathcal X}_{t+1}}\subseteq \bar{\mathcal X}_{t},\;\;\mathcal X_{t+1}\supseteq  \underline{\mathcal X_{t+1}} \supseteq \mathcal X_{t},\\
& \mathcal E_{t+1,i} \subseteq \underline{\mathcal E_{t+1,i} }= \mathcal E_{t,i+1},\;\;\bar{\mathcal E}_{t+1} \subseteq \underline{\bar{\mathcal E}_{t+1}} \subseteq \bar{\mathcal E}_{t},\\
& \widehat{\mathbb X}_{TS_{t+1}} \supseteq \underline{\widehat{\mathbb X}_{TS_{t+1}}} \supseteq \widehat{\mathbb X}_{TS_{t}},\\
&\hat{\mathcal E}_{t+1} \subseteq \underline{\hat{\mathcal E}_{t+1}} \subseteq \hat{\mathcal E}_{t},\;\; \mathbb G_{t+1}\subseteq \underline{\mathbb G_{t+1}}\subseteq \mathbb G_{t}.
\end{align}
\end{subequations}

It is proved that \eqref{nonada}, i.e.,
\begin{subequations} \label{adap_rec_feas}   \begin{align}
    &\mathbb T_{i|t+1}^{\hat x}=\underline{\mathbb T_{i|t+1}^{\hat x}}\;\;\forall i\in\mathbb I_0^{N},\;\;
    \mathbb T_{i|t+1}^{u}=\underline{\mathbb T_{i|t+1}^{u}}\;\;\forall i\in\mathbb I_0^{N-1}\\
    &\text{with }\mathbb G_{t+1}=\underline{\mathbb G_{t+1}}
\end{align}
\end{subequations}
is also a feasible solution for the COCP at time $t+1$ with adaptation limited to the uncertainty sets. 

The tube parameterization satisfies the constraints \eqref{2ndterminala}-\eqref{tubes_sc}, \eqref{cons1m} and \eqref{cons2m} for $\mathbb P_{t+1}$. From \eqref{eq:improved_sets}, it is also observed that
 \begin{subequations}\label{RecFeasTube}
     \begin{align}
 &  \mathbb T_{i|t+1}^{\hat x}  =\underline{\mathbb T^{\hat x}_{i|t+1}} \subseteq \mathbb T^{\hat x^*}_{i+1|t}   \subseteq \widehat{\mathbb X}_{t,i+1}\subseteq \widehat{\mathbb X}_{t+1,i} \;\forall i\in\mathbb I_0^{N-1},   \\
 &  \mathbb T^{\hat x}_{N|t+1}=\underline{\mathbb T^{\hat x}_{N|t+1}}\subseteq \underline{\widehat{\mathbb X}_{TS_{t+1}}}\subseteq \widehat{\mathbb X}_{TS_{t+1}},\\
 & \mathbb T^u_{i|t+1}=\underline{\mathbb T^u_{i|t+1}}\subseteq  \mathbb T^u_{i+1|t} \subseteq\mathbb U \;\;\;\forall i\in\mathbb I_0^{N-2},\\
 &\mathbb T^u_{N-1|t+1}=\underline{\mathbb T^u_{N-1|t+1}} =K_t \underline{\mathbb T^{\hat x}_{N-1|t+1}}\subseteq K_t \mathbb T^{\hat x^*}_{N|t}\subseteq \mathbb U.
     \end{align}
 \end{subequations}
 Further, $\forall j\in\mathbb I_1^{H_{t+1}}$, the vertices $\mathfrak s_{i|t+1}^{[j]}\in\mathbb T^{\hat x}_{i|t+1}= \underline{\mathbb T^x_{i|t+1}}$ and $\mathfrak u_{i|t+1}^{[j]}\in \mathbb T^u_{i|t+1}=\underline{\mathbb T^u_{i|t+1}}$ implies 
\begin{gather}\label{cons6feas}
    \begin{aligned}
        \hat{A}_{t} \mathfrak s_{i|t+1}^{[j]}+\hat{B}_{t} \mathfrak u_{i|t+1}^{[j]} &
   \subseteq \underline{\mathbb{T}^{{\hat x}}_{i+1|t+1}}\ominus \underline{\mathcal E_{t+1,i}} \\
   &\subseteq \mathbb{T}^{{\hat x}}_{i+1|t+1} \ominus \mathcal E_{t+1,i} \;\;\; \forall i\in\mathbb{I}_{0}^{N-1}. 
    \end{aligned}
\end{gather}
Combining \eqref{RecFeasTube} and \eqref{cons6feas}, it is proved that the proposed solution in \eqref{nonada} satisfies all the constraints in $\mathbb P_{t+1}$.
\end{itemize}
To summarize, it is first shown that the COCP $\mathbb P_{t+1}$ might be feasible with a completely new setup; if not, then the backup estimates and sets in \eqref{setup_backup} leading to the proposed solution \eqref{adap_rec_feas} ensure that $\mathbb P_{t+1}$ is feasible. Therefore, if the COCP \eqref{MPC2} with Algorithm~\ref{ch7:algoatube} is feasible at time $t$, then it remains feasible at time $t+1$. By repeating the same argument at time $t+1$ under the updated setup, recursive feasibility follows for all subsequent time instants.

\section{Proof of Theorem \ref{thm:stab} Part (2)}\label{app:thm_on_stab}
Stability is established using the standard MPC Lyapunov argument \cite{kouvaritakis2016model}, by considering the optimal cost of the COCP as a candidate Lyapunov function. 
For the estimated state $\hat x_t$, define
\begin{align*}
    &J_t(\hat x_t,\bar\mu_t)\triangleq\sum_{j=1}^{H_t} \Gamma_t^{[j]},\;\;\text{where }\\
    &\Gamma_t^{[j]}\triangleq \sum_{i=0}^{N-1} \left(\| {\mathfrak{s}}_{i|t}^{[j]}\|^2_Q + \| {\mathfrak{u}}_{i|t}^{[j]}\|^2_R \right)+\| {\mathfrak{s}}_{N|t}^{[j]}\|^2_{P_t} . 
\end{align*}
From $t$ to $t+1$, the tube geometry and consequently the number of vertices change from $H_t$ to $H_{t+1}$. This difference is handled using
\begin{align}
&\tilde \rho_{1_t}\triangleq\begin{cases}
0,\;\text{if } H_{t}\leq H_{t+1}\\
1,\;\text{otherwise,}
\end{cases},\;
\tilde \rho_{2_t}\triangleq\begin{cases}
0,\;\text{if } H_{t}\geq H_{t+1}\\
1,\;\text{otherwise,}
\end{cases}\\
\text{and }&\;\gamma_1\triangleq \max_{t\in\mathbb I_0^\infty,\; \hat x\in\mathbb X\ominus \bar{\mathcal X}_0,\; u\in\mathbb U}  N \left(\left\|\hat x\right\|_Q^2+\left\|u\right\|_R^2 \right)+\left\|\hat x\right\|^2_{P_{t+1}} ,
\label{OF:gamma1}
\end{align}
which is finite by Lemma \ref{lemma:nonfset}, Theorems \ref{theo:observer}, \ref{thm:fixed_rec_feas}, and \ref{thm:stab} Part (1).  

From Appendix~\ref{app:thm:fixed_rec_feas}, the quantities $\underline{\mathfrak s_{i|t+1}^{[j]}}$ and $\underline{\mathfrak u_{i|t+1}^{[j]}}$, $\forall j\in\mathbb I_1^{H_{t+1}}$, represent the vertices of $\underline{\mathbb T^{\hat x}_{i|t+1}}$ and the corresponding elements of $\underline{\mathbb T^u_{i|t+1}}$ defined in \eqref{nonada}, respectively. The worst-case scenario is considered with $\mathbb G_{t+1}=\underline{\mathbb G_{t+1}}$ and $H_{t+1}=\underline{H_{t+1}}$. At any time $t+1$, where $t\in\mathbb I_0^\infty$, it follows that
\begin{align}
&J_{t+1}^*(\hat x_{t+1})\leq J_{t+1}(\hat x_{t+1},\mu_{t+1})=\sum_{j=1}^{H_{t+1}} \Gamma_{t+1}^{[j]}\nonumber\\
&=\sum_{j=1}^{H_{t}} \Gamma_{t+1}^{[j]}-\tilde \rho_{1_t}\sum_{j=H_{t+1}+1}^{H_{t}} \Gamma_{t+1}^{[j]}+\tilde \rho_{2_t}\sum_{j=H_t+1}^{H_{t+1}} \Gamma_{t+1}^{[j]}\nonumber\\ 
&\leq \sum_{j=1}^{H_{t}}\left(\sum_{i=0}^{N-1} \left(\left\|\mathfrak{s}_{i|t+1}^{[j]}\right\|^2_Q + \left\| \mathfrak{u}_{i|t+1}^{[j]}\right\|^2_R \right)+\left\|\mathfrak{s}_{N|t+1}^{[j]}\right\|^2_{P_{t+1}}\right) \nonumber\\
&\;\;\;\;\;\;+\tilde \rho_{2_t}(H_{t+1}-H_t)\gamma_1  \hspace{2cm}(\text{using \eqref{OF:gamma1}})\nonumber\\
&= \sum_{j=1}^{H_{t}}\left(\sum_{i=0}^{N-1} \left(\left\| \underline{\mathfrak{s}_{i|t+1}^{[j]}}\right\|^2_Q + \left\| \underline{\mathfrak{u}_{i|t+1}^{[j]}}\right\|^2_R \right)+\left\|\mathfrak{s}_{N|t+1}^{[j]}\right\|^2_{P_{t+1}}\right) \nonumber\\
&\;\;\;\;\;\;+\tilde \rho_{2_t}(H_{t+1}-H_t)\gamma_1  \hspace{2cm}  \left(\text{using }\eqref{RecFeasTube}\right)  \nonumber\\
&\leq \sum_{j=1}^{H_{t}}\left(\sum_{i=0}^{N-2} \left(\left\| {\mathfrak{s}_{i+1|t}^{[j]^*}}\right\|^2_Q + \left\| {\mathfrak{u}_{i+1|t}^{[j]^*}}\right\|^2_R \right)+\left\|\mathfrak{s}_{N|t+1}^{[j]}\right\|^2_{P_{t+1}} \right.\nonumber\\
&\;\;\;\left.+\left\|\mathfrak{s}_{N|t}^{[j]^*}\right\|^2_{Q+K_t^\top R K_t} \vphantom{\sum_{i=1}^{M}}\right)  +\tilde \rho_{2_t}(H_{t+1}-H_t)\gamma_1 \quad\left(\text{using } \eqref{RecFeasTube}\right)\nonumber\end{align}\begin{align}
&=J_t^*(\hat{x}_t)-\sum_{j=1}^{H_t}\left(\left\| {\mathfrak{s}_{0|t}^{[j]^*}}\right\|^2_Q + \left\| {\mathfrak{u}_{0|t}^{[j]^*}}\right\|^2_R + \left\|{\mathfrak{s}_{N|t}^{[j]^*}}\right\|^2_{P_t} \right)\nonumber\\
&+\sum_{j=1}^{H_t}\left(\left\| (\hat A_{t+1}+\hat B_{t+1}K_{t+1})\mathfrak{s}_{N|t}^{[j]^*} + \varepsilon_{t+1,N-1}\right\|^2_{P_{t+1}}\right.\nonumber\\
&\;\;\;\left.+\left\|\mathfrak{s}_{N|t}^{[j]^*}\right\|^2_{Q+K_t^\top R K_t}\right)  +\tilde \rho_{2_t}(H_{t+1}-H_t)\gamma_1\nonumber\\
&\leq (1-\gamma_2) J_t^*(\hat{x}_t) +\tilde \rho_{2_t}(H_{t+1}-H_t)\gamma_1\nonumber\\
&\;\;\;+\sum_{j=1}^{H_{t}} \left( -\left\|  \mathfrak{s}_{N|{t}}^{[j]^*}\right\|^2_{P_t}+\left\|  \mathfrak{s}_{N|{t}}^{[j]^*}\right\|^2_{Q+K_t^\top R K_t}\right. \nonumber\\
&\;\;\;\left.+\left\|  \mathfrak{s}_{N|{t}}^{[j]^*}\right\|^2_{(\hat A_{t+1}+\hat B_{t+1}K_{t+1})^\top P_{t+1}(\hat A_{t+1}+\hat B_{t+1}K_{t+1})} \right)\nonumber\\
&\;\;\;+H_t \max_{t\in\mathbb I_0^\infty, \;\varepsilon\in\bar{\mathcal E}_t}\left\|\varepsilon \right\|^2_{P_{t+1}}, \;(\text{by Cauchy-Schwartz inequality})\nonumber
\end{align}
where $\gamma_2\triangleq \frac{\sum_{j=1}^{H_t}\left(\left\| {\mathfrak{s}_{0|t}^{[j]^*}}\right\|^2_Q + \left\| {\mathfrak{u}_{0|t}^{[j]^*}}\right\|^2_R\right)}{J_t^*(\hat{x}_t)}\Rightarrow\gamma_2\in(0,1]\Rightarrow 1-\gamma_2\in[0,1)$ (by definition of $J_t^*$).

Irrespective of whether new point estimates satisfying Criterion \ref{Crdare2} are used, or the setup in \eqref{setup_backup} is employed, the following holds:
\begin{align*}
P_t-(\hat A_{t+1}+\hat B_{t+1}K_{t+1})^\top P_{t+1}(\hat A_{t+1}+\hat B_{t+1}K_{t+1})\nonumber\\
-Q-K_t^\top R K_t\succeq 0,
\end{align*}
which implies
\begin{align*}
J_{t+1}^*(\hat x_{t+1})
\leq (1-\gamma_2)J_{t}^*(\hat x_{t})
+\tilde \rho_{2_t}(H_{t+1}-H_t)\gamma_1
\nonumber\\
+H_t \max_{t\in\mathbb I_0^\infty,\;\varepsilon\in \bar{\mathcal E}_t}\left\|\varepsilon \right\|^2_{P_{t+1}}.
\end{align*}

For implementation, an upper bound $H_{\max}$ is imposed on the number of vertices of $\mathbb G_t$. Due to recursive feasibility and bounded sets $\mathbb D$, $\Psi_0$ and $\mathbb X_{0_0}$, $\max_{t\in\mathbb I_0^\infty,\;\varepsilon\in \bar{\mathcal E}_t} \left\|\varepsilon\right\|^2_{P_{t+1}}$ is bounded. The matrices $P_{t+1}$ are also bounded since they appear in the cost function. Therefore, defining
\begin{align*}
\gamma_3\triangleq H_{\max} \max_{t\in\mathbb I_0^\infty,\;\varepsilon\in \bar{\mathcal E}_t}\left\|\varepsilon \right\|^2_{P_{t+1}}+(H_{\max}-H_{\min})\gamma_1,
\end{align*}
yields
\begin{align}
J_{t+1}^*(\hat x_{t+1})\leq (1-\gamma_2) J_{t}^*(\hat x_{t})+\gamma_3.\nonumber
\end{align}
This implies that $\exists$ constants $c_1>0$, $c_2\in(0,1)$ and $c_3>0$ such that
\[
\left\|\hat{x}_t\right\| \leq c_1 c_2^t \left\|\hat{x}_0\right\| + c_3 \quad \forall t\in\mathbb I_0^\infty,
\]
thereby establishing robust exponential stability of the adaptive observer.

Further, once learning has converged, $H_t$ remains constant and the term $H_{t+1}-H_t$ becomes zero. Consequently, the observer state converges to a bounded set determined by $\mathbb D$, $\Psi_t$, $\mathbb X_{0_t}$, and the bounded quantities $P_t$, $K_t$, $\hat A_t$, $\hat B_t$, and $\hat x_{0_t}$.

Since the estimation error is bounded by $\widetilde{\mathbb X}_{0_t}$, which is non-increasing, robust exponential stability of the actual plant \eqref{sys1} also follows.

\section{Proof of Theorem \ref{coro:luenberger}}\label{app:coro:luenberger}
With no change in the point estimates and the uncertainty sets, the observer state can be expressed as
\begin{align*}
    \hat x_{t+1}=M_{t+1}\hat p+F^{t+1}\hat x_0    &=(FM_t+\begin{bmatrix}Y_t & U_t\end{bmatrix})\hat p+ F^{t+1}\hat x_0\\
    &=F\hat x_t + (\hat A_t-F)x_t + \hat B_t u_t,
\end{align*}
where the time indices are not used in $\hat p$ and $\hat x_0$ to represent they are unchanged. Subtracting from \eqref{sys1x} yields
\begin{align}
    \tilde x_{t+1}=F\tilde x_t + (A-\hat A_t)x_t+(B-\hat B_t)u_t +d_t.\label{app:eq:xtilde}
\end{align}

Now, consider the case of a Luenberger observer-based output feedback MPC \cite{mayne2006robust,mayne2009robust,kogel2017robust}, where nominal parameters $\hat A$, $\hat B$ are selected to generate a nominal and error dynamics, and a gain $L_{obs}$ is chosen such that $\hat A- L_{obs} C$ is Schur stable. The system state and observer state are expressed as
\begin{align*}
    &z_{t+1}=\hat A z_t + \hat B u_t +\underbrace{(A-\hat A)z_t +(B-\hat B)u_t + d_t}_{\text{lumped additive disturbance}},\\
   & \hat z_{t+1}=\hat A \hat z_t + \hat B u_t + L_{obs}(Cz_t -C \hat z_t),
   \end{align*}respectively. The observer state estimation error dynamics becomes 
   \begin{align*}
 \tilde z_{t+1}=(\hat A - L_{obs} C)\tilde z_t + (A-\hat A)z_t +(B-\hat B)u_t + d_t,
\end{align*}which is similar to \eqref{app:eq:xtilde} since both $F$ and $(\hat A-L_{obs}C)$ are Schur stable; the dynamics becomes equal provided $F$ is chosen to be equal to $(\hat A-L_{obs}C)$. This also highlights the certainty equivalence principle used in designing adaptive observers. The remaining proof follows from the existing literature \cite{CHISCI20011019,mayne2009robust,kogel2017robust}.

\section*{References}
\bibliographystyle{IEEEtran}
\footnotesize{\bibliography{IEEEabrv,reference}}




\end{document}